\documentclass[12pt,letterpaper,notitlepage]{article}
\usepackage[utf8]{inputenc}
\usepackage[margin=1in]{geometry}
\usepackage{adjustbox}
\usepackage{placeins}
\usepackage{booktabs}
\usepackage{setspace}
\usepackage{amsmath,bm}
\usepackage{amsfonts}
\usepackage{amssymb}
\usepackage{amsthm}
\usepackage[resetlabels]{multibib}
\newcites{Online}{References for Online Appendix}
\usepackage[authoryear]{natbib}
\usepackage{array}
\usepackage{url}
\usepackage{makecell}
\usepackage{hyperref}   
\usepackage{overpic}
\usepackage{xurl}
\usepackage{algorithm2e}
\usepackage{pdflscape}
\usepackage{epigraph}
\usepackage{caption}
\usepackage{subcaption}
\usepackage[disable]{todonotes}
\usepackage{appendix}
\usepackage{enumerate}

\usepackage{tikz}
\usetikzlibrary{positioning}
\usetikzlibrary{snakes}
\usetikzlibrary{calc}
\usetikzlibrary{arrows}
\usetikzlibrary{decorations.markings}
\usetikzlibrary{shapes.misc}
\usetikzlibrary{matrix,shapes,arrows,fit,tikzmark}
\definecolor{red}{RGB}{213,94,0}
\def\sym#1{\ifmmode^{#1}\else\(^{#1}\)\fi}

\tikzset{   
        every picture/.style={remember picture,baseline},
        every node/.style={anchor=base,align=center,outer sep=1.5pt},
        every path/.style={thick},
        }

\tikzstyle{every picture}+=[remember picture] 
\tikzstyle{mybox} =[draw=black, very thick, rectangle, inner sep=10pt, inner ysep=20pt]
\tikzstyle{fancytitle} =[draw=black,fill=red, text=white]

\newcommand{\Rmax}{R_{\max}}

\newcommand{\Rminimax}{{R^*}}
\newcommand{\RBayes}{R_{\operatorname{Bayes}}}
\newcommand{\Aopt}{{A^*}}
\newcommand{\estimator}{\hat\theta}
\newcommand{\bestimator}{\delta}
\newcommand{\deltaadapt}{{\hat\theta^{\ast}}}
\newcommand{\tildedeltaadapt}{\delta^*}
\newcommand{\deltaminimax}{\hat\theta}
\newcommand{\deltaBNM}{\delta}
\newcommand{\rBNM}{r^{\operatorname{BNM}}}

\DeclareMathOperator{\Pen}{Pen}

\newtheorem{lemma}{Lemma}[section]
\newtheorem{theorem}{Theorem}[section]

\theoremstyle{definition}

\newtheorem*{main example}{Main example}
\newtheorem*{main example continued}{Main example (continued)}

\title{\vspace{-3em}
Adapting to Misspecification\thanks{We thank Isaiah Andrews, Manuel Arellano, Dmitry Arkhangelsky, St\'{e}phane Bonhomme, Tom Boot, Xu Cheng, Bryan Graham, Michal Koles\'{a}r, Roger Koenker, Lihua Lei, Jesse Shapiro, and Aleksey Tetenov for helpful discussions on this project. The paper also benefited from the comments of conference and seminar audiences. %
Timothy Armstrong gratefully acknowledges support from National Science Foundation Grant SES-2049765.  Liyang Sun gratefully acknowledges support from the Institute of Education Sciences, U.S. Department of Education, through
Grant R305D200010, and Economic and Social Research Council (new investigator grant UKRI607). Code implementing the adaptive estimators proposed in this paper is available online at \url{https://github.com/lsun20/MissAdapt}.}}
\author{Timothy B. Armstrong\footnote{USC. Email: timothy.armstrong@usc.edu}, Patrick Kline\footnote{UC Berkeley and NBER. Email: pkline@berkeley.edu}\, and Liyang Sun\footnote{UCL and CEMFI. Email: liyang.sun@ucl.ac.uk}}
\date{September 2025}
\begin{document}

\maketitle
\vspace{-3em}
\begin{abstract}

Empirical research typically involves a robustness-efficiency tradeoff. A researcher seeking to estimate a scalar parameter can invoke strong assumptions to motivate a restricted estimator that is precise but may be heavily biased, or they can relax some of these assumptions to motivate a more robust, but variable, unrestricted estimator. When a bound on the bias of the restricted estimator is available, it is optimal to shrink the unrestricted estimator towards the restricted estimator. For settings where a bound on the bias of the restricted estimator is unknown, we propose adaptive estimators that minimize the percentage increase in worst case risk relative to an oracle that knows the bound. We show that adaptive estimators solve a weighted convex minimax problem and provide lookup tables facilitating their rapid computation. Revisiting some well known empirical studies where questions of model specification arise, we examine the advantages of adapting to---rather than testing for---misspecification.\\
\textbf{Keywords:}  Adaptive estimation, Minimax procedures, Specification testing, Shrinkage, Robustness.\\
\textbf{JEL classification codes:} C13, C18.
\end{abstract}

\newpage

\section{Introduction}\label{introduction_sec}

Empirical research is typically characterized by a robustness-efficiency tradeoff. The researcher can either invoke strong assumptions to motivate an estimator that is precise, but sensitive to violations of model assumptions, or they can employ a less precise estimator that is robust to these violations. Familiar examples include the choice of whether to add a set of controls to a regression, whether to exploit over-identifying restrictions in estimation, and whether to allow for endogeneity or measurement error in an explanatory variable.

Decisions of this nature are often approached with a degree of pragmatism: imposing a false restriction may be worthwhile if doing so yields improvements in precision that are not outweighed by corresponding increases in bias. While precision is readily assessed with asymptotic standard errors, the measurement of bias is less standardized.  A popular informal approach is to conduct a series of ``robustness exercises,'' whereby estimates from models that add or subtract assumptions from some baseline are reported and examined for differences. While potentially informative about the presence of bias, it is often unclear how the results of such exercises should be used to refine baseline estimates of the parameter of interest. %

One answer, found often in econometrics textbooks, is
to use a specification test to select a model.
Doing so yields a \emph{pre-test} estimator that equals the estimator of
the restricted model when the specification test fails to reject, and is otherwise
equal to the estimator of the unrestricted model.
The pre-test estimator offers
a form of asymptotic insurance against bias: as the degree of misspecification grows
large relative to the noise in the data, the test rejects with near certainty.
Yet when biases are modest, as one might expect of models that serve as useful
approximations to the world, the cost of this insurance in terms of increased
variance can be exceedingly high.

In this paper we explore an alternative to specification testing:
\emph{adapting} to misspecification.\footnote{An interactive Shiny application implementing our proposed estimator is available online at \url{https://lsun20.github.io/MissAdapt/}.}
Adaptive estimation provides a systematic
approach to exploiting the assumptions of the restricted model as efficiently as possible while acknowledging the
possibility that the restriction in question is misspecified.
Consider an oracle who knows a bound on the extent to which the restricted model is
misspecified, allowing them to combine the estimates from
the restricted and unrestricted models in a way that minimizes maximum risk.
An adaptive estimator is one that comes as close as possible to achieving
this oracle benchmark without using prior knowledge of the magnitude of
misspecification.

We show that adaptive estimators can be computed by solving a
weighted minimax problem.  While the 
resulting 
\emph{optimally adaptive} estimator does not have a closed
form, an analytic soft-thresholding estimator can be tuned to yield comparable performance.
This \emph{adaptive soft-thresholding} estimator can be interpreted as a smoothed version of the pre-test estimator utilizing a critical value that depends on the correlation between the restricted and unrestricted estimators. The near-optimality of adaptive soft-thresholding 
contrasts with the
performance of pre-test estimators, which perform poorly under moderate amounts
of misspecification.

Both the optimally adaptive and adaptive soft-thresholding estimators are easily computed using information that is routinely reported in robustness checks. In the case where the restricted estimator is efficient under the restricted model, the estimators can be computed from published point estimates and standard errors alone.
The adaptive soft-thresholding estimator can also be obtained via a particular sort of lasso regression \citep{tibshirani1996regression} that may be of independent interest in other low-dimensional settings.

To illustrate the advantages of adapting to---rather than testing
for---misspecification, we revisit two empirical examples where questions of
model specification arise.  Our leading example, which we return to throughout the paper, is drawn from \cite{de_chaisemartin_two-way_2020}'s reanalysis of \cite{gentzkow_effect_2011}, in which a two-way fixed effects estimator that exhibits negative weights in many periods is compared to a more variable convex weighted estimator. %
A second example, taken from \cite{angrist_does_1991}, compares an ordinary least squares (OLS) estimate of the returns to schooling to an instrumental variables (IV) estimate. We argue that extra care is required in this example because the IV estimate is orders of magnitude less precise than OLS. 
Online \ref{sec:lalonde} provides an additional example, drawn from \cite{lalonde1986evaluating}, illustrating the problem of estimating the effects of job training using a mix of control groups whose credibility can be ranked ex-ante.
In all of the above examples, adapting between models is found to yield
a more attractive balance between efficiency and robustness %
than
selecting a single model via pre-testing, with the adaptive soft-thresholding estimator performing especially well.

Our analysis builds on early contributions by \citet{hodges_use_1952} and \citet{bickel_minimax_1983,bickel_parametric_1984} who consider families of robustness-efficiency tradeoffs defined over pairs of nested models. We extend this work by considering a continuum of models, indexed by different degrees of misspecification.
A large statistics literature considers the problem of adaptation, defined as the search for an estimator that performs nearly as well as an oracle with additional knowledge of the data generating process. We focus on the case where proximity to oracle performance is measured in terms of the ratio of actual to oracle risk, which mirrors the definition used in \citet{tsybakov_pointwise_1998} and leads to simple risk guarantees and statements about relative efficiency. To introduce the core ideas, we begin with a simple introductory example.

\section{An introductory example}\label{sec:intro_example}

In this section, we 
illustrate our proposal at a high level via an empirical example, postponing the details to later discussion. \cite{gentzkow_effect_2011} studied the effects of newspapers on voter turnout in US presidential elections using a two-way fixed effects (TWFE) model estimated in first differences by least squares. 
\cite{de_chaisemartin_two-way_2020} showed that in settings featuring staggered adoption, like the one studied by \cite{gentzkow_effect_2011}, TWFE estimators %
identify potentially non-convex combinations of average treatment effects over time and across adoption cohorts.

Suppose the target parameter $\theta$ is the average effect of changing newspaper access on voter turnout in counties exhibiting a change in the number of newspapers. Let $Y_R$ denote the TWFE estimator used by \cite{gentzkow_effect_2011}
and $Y_U$ the 
estimator of $\theta$ proposed by \cite{de_chaisemartin_two-way_2020}. In the presence of treatment effect heterogeneity, $Y_R$ likely identifies a different parameter, implying an unknown bias $b=E[Y_R]-\theta.$ In contrast, $Y_U$ is unbiased for $\theta$. However, when treatment effect heterogeneity is mild, $Y_R$ may exhibit negligible bias and substantially lower variance than $Y_U$, yielding a non-trivial robustness-efficiency tradeoff.

The value of $Y_R$ reported by \cite{gentzkow_effect_2011} implies that an additional newspaper raises voter turnout by $0.26$ percentage points, with a standard error of $\sigma_R=0.09$. The unrestricted estimator $Y_U$
evaluates to 0.43, with a
standard error of $\sigma_U=0.14$.  Suppose that $Y_U$ and $Y_R$ are
normally distributed with standard deviations given by these standard errors, an
approximation that can be formally justified using a local asymptotic
misspecification framework.  
The difference $Y_O=Y_R-Y_U$  gives a noisy estimate of the bias $b$. To further simplify the example, suppose that
$\operatorname{cov}(Y_R,Y_O)=0$. This condition, which seems to be very nearly satisfied in the data, implies that $Y_R$ is efficient under the
constraint $b=0$.
Consequently, the variance of $Y_O$ is  $\sigma_O^2=\sigma_U^2-\sigma_R^2.$ 
The  test statistic  that forms the basis for standard ``over-identification'' tests of specification is $T_O=Y_O/\sigma_O.$

To compare these estimators, consider their mean squared error (MSE), which will be our preferred measure of risk.  Since $Y_U$ is unbiased, its MSE is equal to its variance
$\sigma_U^2=(0.14)^2$.  In contrast, the MSE of the
restricted estimator depends on its bias $b$:
$E[(Y_R-\theta)^2]=b^2+\sigma_R^2=b^2+(0.09)^2$.  
Figure \ref{fig:build_oracle} plots the MSE of the unrestricted and restricted estimators as functions of the unknown bias $b$. To ease visual interpretation both risk functions have been divided by
$\operatorname{var}(Y_U)$, which normalizes the risk of $Y_U$ to $1$.

\begin{figure}[h!]
\caption{\label{fig:build_oracle}Risk of unrestricted, restricted, $B$-minimax, and oracle estimators}
\begin{centering}
\includegraphics[width=4.5in, trim={0 1.5cm 0 0.5cm}, clip]
{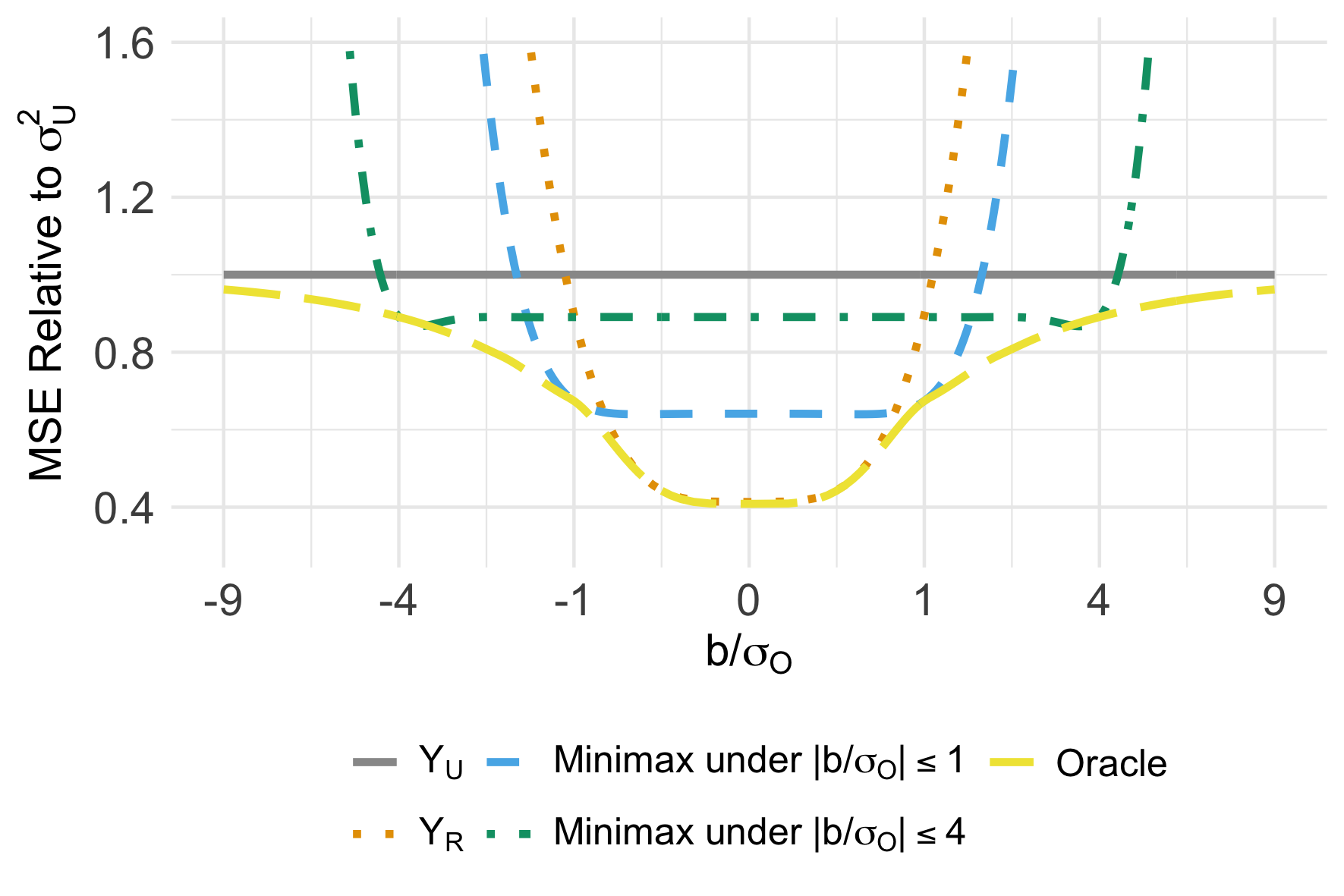}

\footnotesize{Notes: Depiction assumes $\sigma_R^2/\sigma_U^2=0.41$. Horizontal axis is spaced quadratically.}
\par\end{centering}
\end{figure}

When $b=0$, using $Y_R$ instead of $Y_U$ yields a
decrease in MSE from $(0.14)^2$ to $(0.09)^2$.  
The price paid for this improvement in MSE at $b=0$ is that the MSE can be
much larger than  $(0.14)^2$
when $b\ne 0$.
Tradeoffs of this nature are unavoidable because $Y_U$ is admissible: no other estimator has lower MSE for all $b$.
The goal of adaptive estimation is to resolve this tradeoff by balancing efficiency when $b$ is close to zero against robustness when $b$ is large.

Given a bound $B\ge 0$ on the bias magnitude $|b|$, one can compute the estimator that is minimax over the restricted parameter space $(\theta,b)\in \mathbb{R}\times [-B,B]$, a procedure
we refer to as the \emph{$B$-minimax estimator}.
The $B$-minimax estimator $\deltaminimax(Y_U,T_O;B)=Y_R - \sigma_O \deltaBNM(T_O;B)$ adjusts the restricted estimator by an estimate $\deltaBNM(T_O;B)$ of its bias constructed by smoothly shrinking $T_O$ towards zero, yielding output in the interval $[-B/\sigma_O,B/\sigma_O].$
Figure \ref{fig:build_oracle} plots the risk function of the
$B$-minimax estimator for $B\in\{\sigma_O,4\sigma_O\}$. As a benchmark, we also plot the risk function of an \emph{oracle estimator} computed using prior knowledge of the best possible bound $B=|b|$.

Note that if the posited bound $B$ is set lower than the true bias magnitude $|b|$, $B$-minimax estimation can yield
very large MSE.
An alternative to guessing a bound $B$ is to use the data to infer a likely value of $|b|$. Then one can estimate $\theta$ optimally subject to the estimated bias magnitude. The pre-test estimator described in the introduction uses $Y_U$  when $|Y_O|>1.96 \sigma_O$ and otherwise relies on $Y_R$. 
Unfortunately, the
risk function of the pre-test estimator, plotted in Figure \ref{fig:pre_test}, is quite large for moderate values of $b$, reflecting the cost of using the data ``twice'' in a non-smooth fashion. 

\begin{figure}[h!]
\caption{\label{fig:pre_test}Risk of optimally adaptive, soft-thresholding, and pre test estimators}
\begin{centering}
\includegraphics[width=4.5in, trim={0 1.5cm 0 0.5cm}, clip]{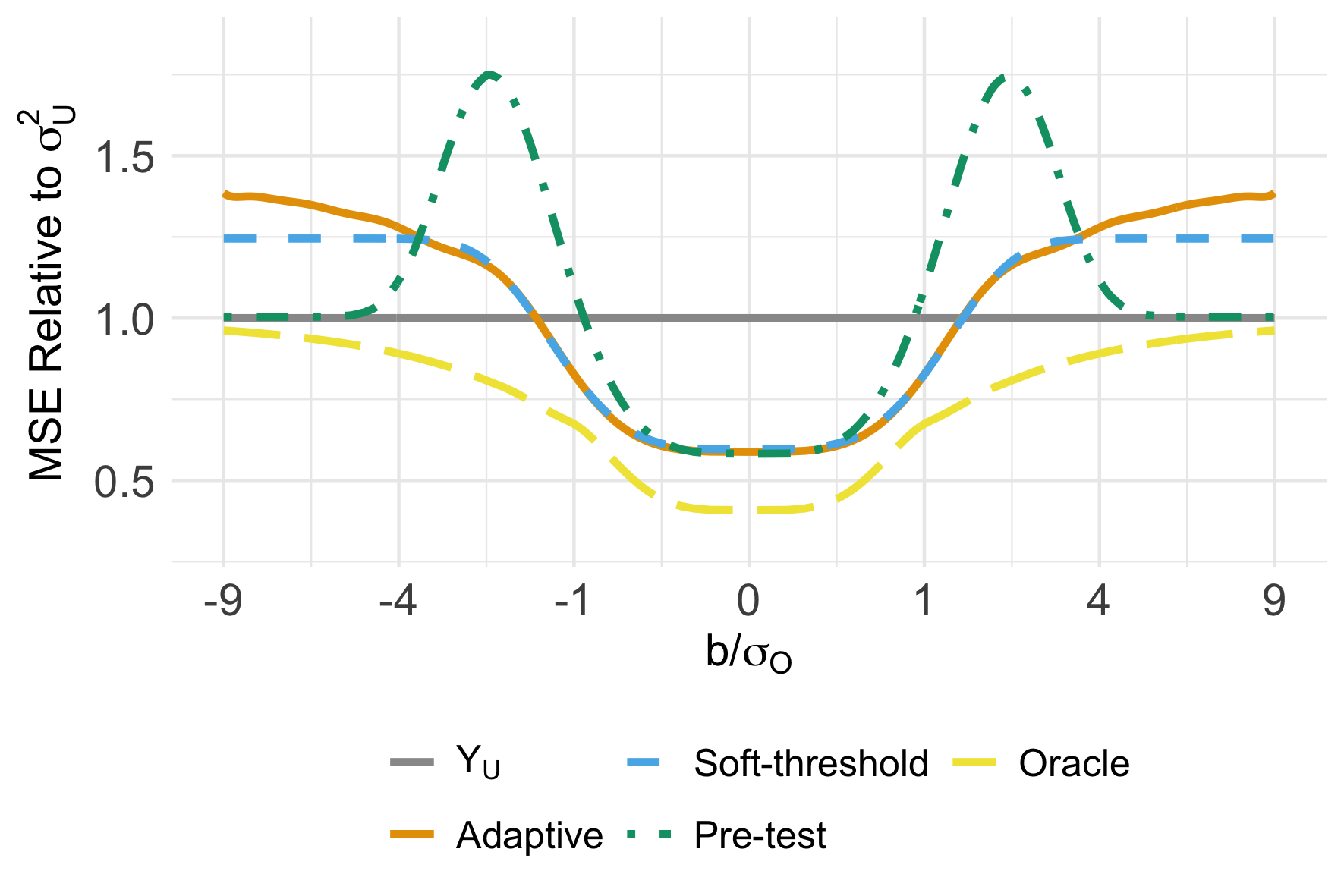}

\footnotesize{Notes: Depiction assumes $\sigma_R^2/\sigma_U^2=0.41$. Horizontal axis is spaced quadratically.}
\par\end{centering}

\end{figure}

Adaptive estimators, by contrast, use the data to directly mimic the oracle's risk function.
 The \emph{optimally
  adaptive estimator} is the estimator that comes closest to matching the oracle's risk function, where distance is measured in terms of the maximum ratio of actual to oracle risk across all bias levels, a metric that we term the \emph{adaptation regret}.  
Like the $B$-minimax estimator, the optimally adaptive estimator $\deltaadapt(Y_U,T_O)=Y_R - \sigma_O \tildedeltaadapt(T_O;1-\sigma_R^2/\sigma_U^2)$ uses $T_O$ to adjust the restricted estimator for bias; however, it depends on $\sigma_R^2/\sigma_U^2$, which captures the efficiency of $Y_U$ relative to $Y_R$, rather than on an ex-ante bound $B$. Though the function $\tildedeltaadapt(\cdot;1-\sigma_R^2/\sigma_U^2)$ lacks an analytic closed form, 
a simple soft-thresholding estimator can be tuned to approximate it closely.

Like the pre-test estimator, the resulting \emph{adaptive soft-thresholding estimator} is equal to
$Y_R$ if $|Y_O/\sigma_O|$ is less than some threshold value $\lambda$.  However,
rather than switching discontinuously to $Y_U$ when $|Y_O|>\lambda \sigma_O$, the
soft-thresholding estimator ``shrinks'' the unrestricted estimator towards the restricted estimator by $\lambda$ standard errors of the bias estimate.
The optimal threshold is a decreasing function of the ratio $\sigma_R^2/\sigma_U^2$. 
In this example, $\sigma_R^2/\sigma_U^2= 0.41$, implying $Y_U$ is only 41\% as efficient as $Y_R$ when $b=0$. The corresponding optimal threshold 
is $\lambda=0.64$, far below the traditional 1.96 value used for pre-testing.

The risk function of the optimally adaptive estimator and its soft-thresholding
approximation are shown in Figure \ref{fig:pre_test}.  The MSE of the optimally adaptive
estimator is never more than 44\% above the oracle MSE, which is the best that
can be achieved. The adaptive soft-thresholding estimator has an MSE that is never more than 46\% above the oracle. When $b=0$, these adaptive estimators achieve substantially %
lower MSE than $Y_U$. Conversely, when $|b|$ is large, they exhibit modestly 
higher MSE than $Y_U$.  The pre-test estimator also achieves near oracle MSE levels when $b=0$. However, when $|b|\approx1.96\sigma_O$, its MSE is 118\% percent
above the oracle MSE and 75\% above the MSE of $Y_U$. 

\begin{figure}[h!]
\caption{\label{fig:prior} Least favorable priors when $\sigma_R^2/\sigma_U^2=0.41$}%
\begin{centering}
\includegraphics[width=4.5in, trim={0 2cm 0 0.5cm}, clip]{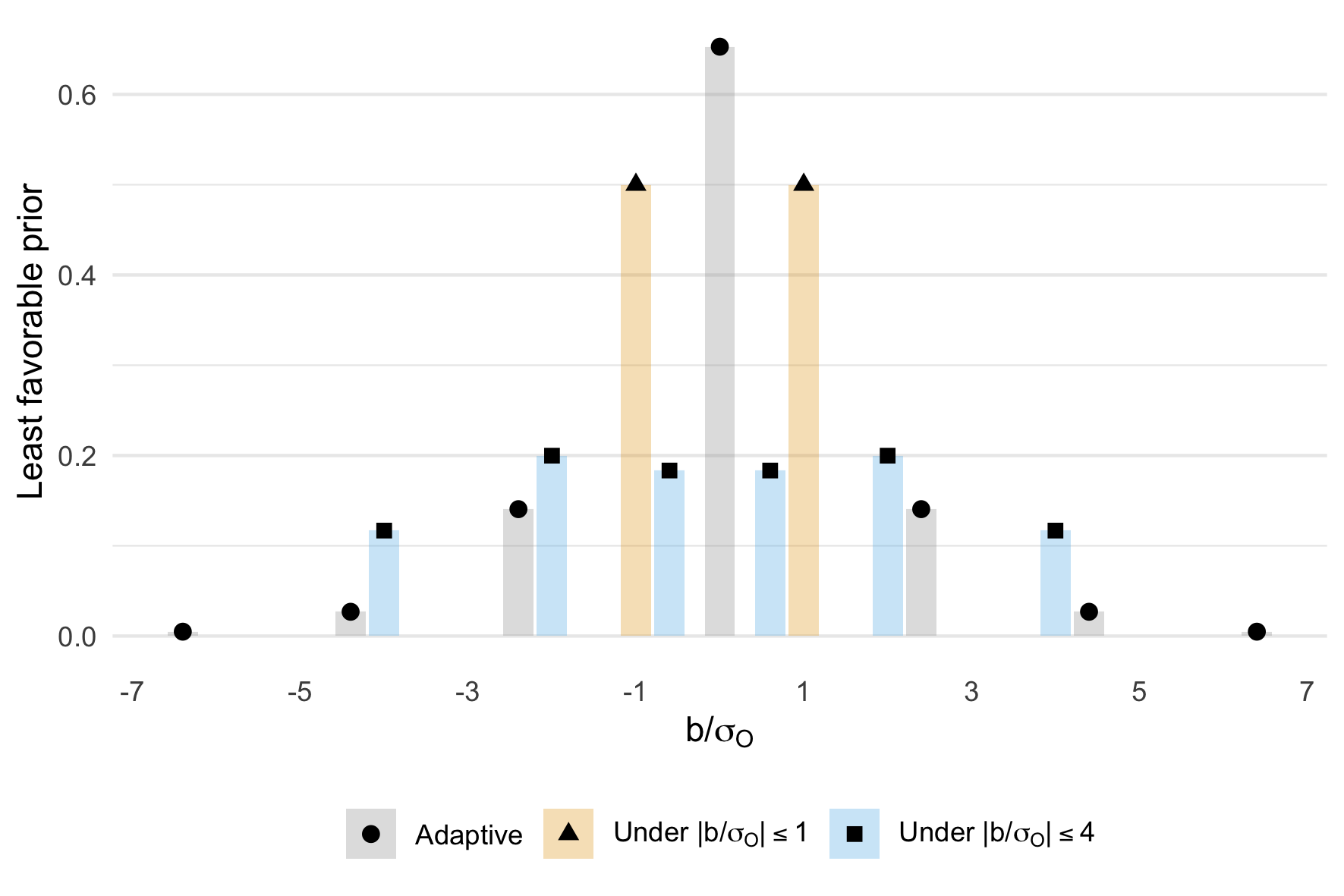}
\par\end{centering}
\end{figure}

Both the
adaptive estimator and its $B$-minimax counterparts can %
be thought of as Bayes
estimators motivated by particular least favorable priors. Figure
\ref{fig:prior} depicts the least favorable priors utilized by the $B$-minimax
estimator for two values of $B$ along with the least favorable prior of the
adaptive estimator.
All three priors
are discrete, symmetric about zero, and decreasing in $|b|$. The $B$-minimax priors have support on $[-B/\sigma_O,B/\sigma_O]$ but involve more than two mass points when $B$ is large, which yields statistical uncertainty about both the magnitude and sign of the bias. 

The estimators motivated by these three priors will all tend to yield lower MSE than $Y_U$ when the true bias magnitude $|b|$ is small. 
The adaptive prior has the important advantage over $B$-minimax priors of not requiring specification of the bound $B$. 
Moreover, the adaptive prior is \emph{robust}: the risk of the optimally adaptive estimator remains bounded as $|b|$ grows large, whereas the risk of a $B$-minimax estimator grows without limit once $|b|$ exceeds the posited bound $B$.

\section{Setup}\label{preliminaries_sec}

Consider a researcher who observes data or initial estimate $Y$ taking values in a set $\mathcal{Y}$, following a distribution $P_{\theta,b}$ that depends on unknown parameters $(\theta,b)$. 
 Let $E_{\theta,b}$ denote expectation under the distribution $P_{\theta,b}$.  We will study possibly misspecified models in a normal or asymptotically normal setting. Results covering more general models are available in a prior version of this paper \citep{armstrong2023adapting}.

The random variable $Y=(Y_U,Y_R)$ consists of an  unrestricted estimator $Y_U$ of a scalar parameter $\theta\in\mathbb{R}$ and a restricted estimator $Y_R$ that is predicated upon additional model assumptions. The additional restrictions required to motivate the restricted estimator make it less robust but potentially more efficient. To capture this tradeoff, we assume that $Y_U$ is asymptotically unbiased for $\theta$, while $Y_R$ may exhibit a bias of $b$ stemming from violation of the additional restrictions. We focus on the case where $Y_R$ is a single scalar-valued estimate, but extensions to vector-valued $b$ are provided in Appendix \ref{details_and_proofs_appendix}.

It will often be convenient to work with the quantity $Y_O=Y_R-Y_U$, which gives an estimate of the bias $b$ that features in conventional tests of over-identifying restrictions.
We work with the large sample approximation
\begin{align}\label{main_example_eq}
  \left(  \begin{array}{c}
            Y_U  \\
            Y_O
          \end{array}\right)
  \sim N\left( \left(
  \begin{array}{c}
    \theta  \\
    b
  \end{array}
  \right),
   \Sigma
  \right),
  \quad
  \Sigma = \left( \begin{array}{cc}
             \sigma_{U}^2 & \rho \sigma_U \sigma_O  \\
             \rho \sigma_U \sigma_O & \sigma_{O}^2
           \end{array} \right).
\end{align}
The variance matrix $\Sigma$ is treated as known.  In practice, feasible versions of our procedures can be computed using a consistent estimate of the asymptotic variance matrix.  The model (\ref{main_example_eq}) arises as from a local asymptotic framework where $\theta$ and $b$ are scaled by the square root of the sample size and $Y_U$ and $Y_R$ are asymptotically normal.  

Under the restriction $b=0$, the efficient GMM estimator of $\theta$ is $Y_{R,GMM}$ and its variance is $\sigma^2_{R,GMM}$, where
\begin{align}\label{eff_gmm_eq}
    Y_{R,GMM}:=Y_U-(\rho\sigma_U/\sigma_O) Y_O,
    \quad \sigma_{R,GMM}^2:=\operatorname{var}(Y_{R,GMM})=\sigma_U^2\cdot (1-\rho^2).
\end{align}
In the case where $\rho \sigma_U \sigma_O=-\sigma_O^2$, the restricted estimator $Y_R$ and the efficient GMM estimator $Y_{R,GMM}$ coincide because $\operatorname{cov}(Y_R,Y_O)=0$.  One can easily compute $\sigma_O^2$ in this case from the contrast $\sigma_U^2-\sigma_R^2$ \citep{hausman_specification_1978}. Likewise, when $Y_U$ and $Y_R$ are estimated on independent samples, computation is facilitated by the simple relation $\sigma_O^2=\sigma_R^2+\sigma_U^2.$ The \emph{relative efficiency} of $Y_U$ to $Y_{R,GMM}$ is given by $\sigma_{R,GMM}^2/\sigma_U^2=1-\rho^2$.

\subsection{$B$-minimax estimators}

An estimator $\estimator:\mathcal{Y}\to\mathcal{A}$ maps the data $Y$ to an action
$a\in\mathcal{A}$. The loss of taking action $a$ under parameters $(\theta,b)$ is given by the function $L(\theta,b,a)$. While it is possible to analyze many types of loss functions in our framework, we will focus on the familiar case of estimation of a scalar parameter $\theta\in\mathbb{R}$ with $\mathcal{A}=\mathbb{R}$ and squared
error loss $L(\theta,b,\estimator)=(\estimator-\theta)^2$.

The risk of an estimator is given by the function
\begin{align*}
  R(\theta,b,\estimator)=E_{\theta,b}L(\theta,b,\estimator(Y))
  =\int L(\theta,b,\estimator(y)) \, d P_{\theta,b}(y).
\end{align*}
An estimator $\estimator$ is \emph{minimax} over the set $\mathcal{C}$ for the parameter $(\theta,b)$ if it minimizes the maximum risk over $(\theta,b)\in\mathcal{C}$.  We are interested in a setting where the researcher entertains multiple parameter spaces $\mathcal{C}_B$, indexed by $B\in\mathcal{B}$, which may restrict the parameters $(\theta,b)$ in different ways.  Define the \emph{$B$-minimax} estimator as the $\hat \theta$ that is minimax over $\mathcal{C}_B$ and its maximum risk $R^*(B)$ as the \emph{$B$-minimax risk}:
\begin{align*}
  \Rminimax(B)
  = \inf_{\estimator} \Rmax (B,\estimator) \quad\text{where}\quad \Rmax(B,\estimator)=\sup_{(\theta,b)\in \mathcal{C}_B}R(\theta,b,\estimator).
\end{align*}

We will focus on the parameter spaces:
\begin{align*}
  \mathcal{C}_B=\{(\theta,b): \theta\in\mathbb{R}, b\in [-B,B]\} = \mathbb{R}\times [-B,B] 
\end{align*}
indexed by a scalar bound $B$ on the magnitude of the bias of the restricted estimator.
Hence, the set $\mathcal{C}_\infty$ corresponds to the unrestricted parameter space, while $\mathcal{C}_0$ corresponds to the restricted parameter space.  Consequently, the $\infty$-minimax estimator (the $B$-minimax estimator when $B=\infty$) is $Y_U$, while the $0$-minimax estimator (the $B$-minimax estimator when $B$=0) is 
$Y_{R,GMM}$. In the special case where the restricted estimator is fully efficient, the $0$-minimax estimator is additionally equal to the restricted estimator $Y_R=Y_U+Y_O$.

\subsection{Adaptation}

Researchers are often unwilling to commit to a restricted parameter space $\mathcal{C}_B$, either because they lack appropriate prior information or because priors differ among their scientific peers. 
Relative to an oracle that knows $|b|\leq B$ and is able to compute the $B$-minimax estimator, an estimator $\estimator$ formed without reference to a particular parameter space $\mathcal{C}_B$ yields a proportional increase in worst-case risk given by
\begin{align*}
  &A(B,\estimator) 
    =\frac{\Rmax(B,\estimator)}{\Rminimax(B)}.
\end{align*}
We refer to $A(B,\estimator)$ as the \emph{adaptation regret} of the estimator $\estimator$ under the set $\mathcal{C}_B$. In our main results, risk corresponds to mean squared error. Hence, $(A(B,\estimator)-1)\times100$ gives the percentage increase in worst-case MSE over $\mathcal{C}_B$ faced by an estimator $\estimator$ relative to the $B$-minimax estimator.

Define the \emph{worst case adaptation regret}
as $A_{\max}(\mathcal{B},\estimator)=\sup_{B\in\mathcal{B}}A(B,\estimator)$.
The lowest possible value $A_{\max}(\mathcal{B},\estimator)$ can take is
\begin{align}\label{Aopt_def_eq}
  \Aopt(\mathcal{B})
  = \inf_{\estimator} \sup_{B\in\mathcal{B}} A(B,\estimator) 
  = \inf_{\estimator}\sup_{B\in\mathcal{B}} \frac{\Rmax(B,\estimator)}{\Rminimax(B)}.
\end{align}
Following \citet{tsybakov_pointwise_1998}, $\Aopt(\mathcal{B})$
gives the \emph{loss of efficiency under adaptation}. An estimator $\estimator$ is \emph{optimally adaptive} if $A_{\max}(\mathcal{B},\estimator)=\Aopt(\mathcal{B})$. We use the symbol $\deltaadapt$ to represent such an estimator. 

Note that different ways of defining adaptation regret---e.g., in terms of the level increase in risk, rather than the proportional increase---would lead to different optimally adaptive estimators. 
The proposed definition has the important advantage of being scale invariant: a change of the units in which MSE is measured will not alter the percentage increase in risk over an oracle. However, when $R^*(0)$ is near zero, the optimally adaptive estimator will prioritize minimizing MSE under $b=0$, a difficulty addressed in Section \ref{constrained_adaptation_sec}.

We study parameter spaces $\mathcal{C}_B=\mathbb{R}\times [-B,B]$, where the set of values of $B$ under consideration is $\mathcal{B}=[0,\infty]$. 
Adaptive estimators yield worst case risk near $R^*(B)$ for all $B$, thereby avoiding commitment to a particular choice of $B$.
Another way to avoid specifying $B$ is to make the conservative choice $B=\infty$, leading to the $\infty$-minimax estimator $Y_U$.
Since $Y_U$ is admissible, the optimally adaptive estimator cannot provide a uniform improvement on $Y_U$ for all $b\in \mathbb{R}$.
However, the optimally adaptive estimator does a better job of mimicking the $B$-minimax estimator for small $B$, while also limiting the increase in risk over $Y_U$ in the worst case.

Early work by \cite{bickel_parametric_1984} considered adapting over the granular set $\mathcal{B}^{gran}=\{0,\infty\}$. Naturally, it is easier to adapt to the elements of the finite set $\mathcal{B}^{gran}$ than to the infinite set $\mathcal{B}$. Consequently, $\Aopt(\mathcal{B}^{gran})\leq\Aopt(\mathcal{B})$. However, consideration of $\mathcal{B}^{gran}$ may leave efficiency gains on the table for $0<b<\infty$ because $R^*(b) \leq R^*(\infty)$.

In \ref{appdx:group} we develop a stylized model that illustrates the ability of adaptive decisions to foster consensus among ``committees'' characterized by different sets of beliefs. 
When the loss of efficiency under adaptation $\Aopt(\mathcal{B})$ is not too large, the committees will agree to jointly follow the optimally adaptive decision because every committee can be compensated for the small increase in maximum risk over their preferred $B$-minimax level. Taking the committees to represent different camps of researchers, the model suggests adaptive estimation can help to forge consensus between researchers with varying beliefs about the suitability of different econometric models.

\section{Main results}\label{main_results_sec}

This section derives the form of the optimally adaptive estimator in our
setting.  We begin by noting that the problem of computing adaptive estimators
can be reduced to that of computing minimax estimators with a scaled loss function.

\subsection{Adaptation as minimax with scaled
  loss}\label{adaptation_as_scaled_minimax_sec}

Plugging in the definition of $\Rmax(B,\estimator)$ along with
$\mathcal{B}=[0,\infty]$ and $\mathcal{C}_B=\mathbb{R}\times[-B,B]$, the
criterion that the optimally adaptive estimator $\deltaadapt$ minimizes can be written
\begin{align*}
\sup_{B\in [0,\infty]} \frac{\Rmax(B,\estimator)}{\Rminimax(B)}
  =\sup_{B\in [0,\infty]}\sup_{\theta\in\mathbb{R},b\in [-B,B]}\frac{R(\theta,b,\estimator)}{\Rminimax(B)}
  =\sup_{(\theta,b)\in\mathbb{R}^2} \sup_{B\in [|b|,\infty]} \frac{R(\theta,b,\estimator)}{\Rminimax(B)}
\end{align*}
where the last equality follows by noting that the double supremum on either
side of this equality is over the same set of values of $(B,\theta,b)$.
Since $R^*(B)$ is increasing in $B$, the inner supremum is taken at $B=|b|$,
which gives the following lemma.

\begin{lemma}\label{adaptation_as_weighted_minimax_lemma}
The loss of efficiency under adaptation (\ref{Aopt_def_eq}) is given by
\begin{align*}
  \inf_{\estimator}\sup_{(\theta,b)\in \mathbb{R}^2} \omega(b) R(\theta,b,\estimator)
  \quad\text{where}\quad \omega(b) = 1/R^*(|b|)
\end{align*}
and an estimator $\deltaadapt$ that achieves this infimum (if it exists) is optimally adaptive.
\end{lemma}

Lemma \ref{adaptation_as_weighted_minimax_lemma} shows that finding an optimally adaptive decision can be written as a minimax problem with a weighted version of the original loss function.  In particular, $\deltaadapt$ is found to minimize the maximum (over $\theta,b$) of the objective $\omega(b)R(\theta,b,\estimator)=E_{\theta,b}\omega(b)L(\theta,b,\estimator(Y))$. 
Hence, the optimal adaptive estimator corresponds to a minimax estimator under the loss function $\omega(b)L(\theta,b,\estimator(Y))$.

\subsection{$B$-minimax and adaptive estimators}\label{main_example_computation_sec}

According to Lemma \ref{adaptation_as_weighted_minimax_lemma}, computing
adaptive estimators amounts to solving a weighted minimax problem.  In our
setting, we can further simplify this problem using invariance.
We focus here on the case of squared error loss $L(\theta,b,\estimator)=(\theta-\estimator)^2$.
Appendix \ref{details_and_proofs_appendix} provides proofs of the results in this section and covers general loss functions for estimation of the form $L(\theta,b,\estimator)=\ell(\theta-\estimator)$. It will be useful to transform the data to $(Y_U,T_O)$, where
$T_O=Y_O/\sigma_O$ is the $t$-statistic for a specification test of the
null that $b=0$.
This representation is equivalent to our original setting because $\sigma_O$ is known. %

It follows from invariance arguments that both the $B$-minimax estimator and the optimally adaptive estimator take the form
\begin{align}\label{Bminimax_main_example_eq}
\estimator(Y_U,T_O)=\rho\sigma_U \bestimator\left(T_O\right) + Y_U - \rho \sigma_U T_O
=\rho\sigma_U \bestimator\left(T_O\right) + Y_{R,GMM},
\end{align}
where $Y_{R,GMM}$ is the efficient GMM estimator given in (\ref{eff_gmm_eq}) and $\delta:\mathbb{R}\rightarrow\mathbb{R}$ is an estimator of the scaled bias $b/\sigma_O$.
Note that,
when $b\ne 0$, $Y_{R,GMM}$ exhibits a bias of $-(\rho\sigma_{U}/\sigma_{O}) b$.  
The estimator in (\ref{Bminimax_main_example_eq}) subtracts from the GMM estimate a corresponding estimate $-\rho\sigma_{U} \bestimator\left(Y_O/\sigma_O\right)$ of this bias term. Estimators in this class were also considered by \citet{magnus_estimation_1999} in the context of linear regression.

The following theorem, which is proved in Appendix \ref{details_and_proofs_appendix}, describes the particular functions $\bestimator(\cdot)$ in the class of estimators defined by (\ref{Bminimax_main_example_eq}), used by the $B$-minimax and optimally adaptive estimators.

\begin{theorem}\label{main_example_thm_main_text}
Consider the model in (\ref{main_example_eq}) with parameter spaces $\mathcal{C}_B=\mathbb{R}\times [-B,B]$ for $B\in\mathcal{B}=[0,\infty]$ and squared error loss $L(\theta,b,\estimator)=(\estimator-\theta)^2$. The following results hold:
\begin{enumerate}[(i)]
    \item\label{main_example_thm_minimax_estimator}
    Let $\deltaBNM\left(T_O;B\right)$ be the minimax estimator of $\vartheta \in \mathcal{C}=[-B/\sigma_O,B/\sigma_O]$ when $T_O\sim N(\vartheta, 1)$ and let $\rBNM \left( B/\sigma_O \right)$ be the corresponding minimax risk.
    The $B$-minimax estimator of $\theta$ is given by
    \begin{align*}
        \deltaminimax(Y_U,T_O;B)=\rho\sigma_U \deltaBNM\left(T_O;B\right) + Y_U - \rho \sigma_U T_O
    \end{align*}
    and the $B$-minimax risk is given by
    \begin{align}\label{Rminimax_main_example_eq}
\Rminimax(B) =\rho^2 \sigma_U^2\rBNM \left( B/\sigma_O \right)
  + \sigma_U^2 - \rho^2 \sigma_U^2.
    \end{align}

    \item\label{main_example_thm_adaptive_estimator}
    An optimally adaptive estimator of $\theta$ takes the form
    \begin{align*}
    \deltaadapt(Y_U,T_O)=\rho\sigma_U \tildedeltaadapt\left(T_O;\rho^2\right) + Y_U - \rho \sigma_U T_O,
    \end{align*}
    where $\tildedeltaadapt(\cdot;\rho^2)$ is a function that minimizes
    \begin{align}\label{adaptation_objective_main_example_eq}
\sup_{\tilde b\in\mathbb{R}} \frac{E_{T\sim N(\tilde b,1)}(\bestimator(T)-\tilde b)^2 + \rho^{-2} - 1}{\rBNM(|\tilde b|) + \rho^{-2}-1}.
\end{align}

    \item\label{main_example_thm_lea}
    The loss of efficiency under adaptation $\Aopt(\mathcal{B})$ in \eqref{Aopt_def_eq} is equal to 
\begin{align*}
\inf_{\bestimator} \sup_{\tilde b\in\mathbb{R}} \frac{E_{T\sim N(\tilde b,1)}(\bestimator(T)-\tilde b)^2 + \rho^{-2} - 1}{\rBNM(|\tilde b|) + \rho^{-2}-1}
=\sup_{\pi} \inf_{\bestimator} \int \frac{E_{T\sim N(\tilde b,1)}(\bestimator(T)-\tilde b)^2 + \rho^{-2} - 1}{\rBNM(|\tilde b|) + \rho^{-2}-1}\, d\pi(\tilde b)
\end{align*}
where the supremum is over all probability distributions $\pi$ on $\mathbb{R}$.
\end{enumerate}

\end{theorem}

Part (\ref{main_example_thm_minimax_estimator}) of Theorem \ref{main_example_thm_main_text} establishes that  the $B$-minimax estimator relies on an estimator $\deltaBNM(T_O;B)$ of the scaled bias $b/\sigma_O$ that is minimax under the bound $|b|\le B$.  This minimax estimation problem is called the \emph{bounded normal mean} problem and has been studied extensively in the literature.
We detail the computation of this estimator in Online Appendix \ref{appdx: lookup_bnm}. For finite $B/\sigma_O$, the minimax estimator is the posterior mean against a least favorable prior. Figure \ref{fig:prior} illustrates several such priors. When the interval is small, the least favorable prior concentrates at the two endpoints. For larger intervals, it concentrates at a finite number of points within $[-B/\sigma_O, B/\sigma_O]$ \citep{casella_estimating_1981}. For $B/\sigma_O = \infty$, the minimax estimator is $T_O$.

Theorem \ref{main_example_thm_main_text}(\ref{main_example_thm_adaptive_estimator}) states that the optimally adaptive
estimator takes the form in (\ref{Bminimax_main_example_eq}) with $\bestimator(\cdot)$ given by $\tildedeltaadapt(\cdot;\rho^2)$: the solution to a
weighted minimax problem over the scaled bias
$\tilde b=b/\sigma_O$. Following part (iii) of Theorem
\ref{main_example_thm_main_text},  the problem is solved
numerically using a discrete approximation to the
least favorable
prior over $\tilde
b$
as in \citet{chamberlain_econometric_2000}.
The least favorable prior distributions reported in Figure \ref{fig:prior} were
computed using this approach.
The invariance arguments used to derive
(\ref{adaptation_objective_main_example_eq}) imply an independent flat prior for $\theta$. 
To streamline computation, $\tildedeltaadapt(\cdot; \rho^2)$ is evaluated on a grid of $\rho^2$ values, creating a lookup table. See Online \ref{appdx: lookup} for details.

One can write the optimally adaptive estimator as a weighted average:
\begin{align}
\deltaadapt(Y_U,T_O) = w(T_O)\cdot Y_U + \left(1-w(T_O)\right) \cdot Y_{R,GMM},\nonumber
\end{align}
where $w(T_O)=\tildedeltaadapt(T_O;\rho^2)/T_O$ is a data-dependent weight. %
We find numerically that the adaptive estimator ``shrinks'' $T_O$ towards zero, leading the weight $w(T_O) $ to fall between zero and one for all values of $\rho^2$. 
The data dependent nature of the weight $w(T_O)$ is clearly crucial for the robustness properties of the optimally adaptive estimator. As $T_O$ grows large, less weight is placed on the optimal GMM estimator and more weight is placed on the unrestricted estimator $Y_U$. If one were to commit ex-ante to a fixed (i.e., non-stochastic) weight on $Y_U$ below one, the worst-case risk of the procedure would become unbounded because the optimal GMM estimator can exhibit arbitrarily large bias.

\subsubsection{Impossibility of consistently estimating the asymptotic distribution}

The distribution of an estimator of the form (\ref{Bminimax_main_example_eq}) can be derived by noting that $Y_{R,GMM}$ and $T_O$ are independent, with $Y_{R,GMM}\sim N(\theta-b \rho \sigma_U/\sigma_O, \sigma^2_U(1-\rho^2))$ and $T_O\sim N(b/\sigma_O, 1)$. Let $Z_1$ and $Z_2$ denote independent $N(0,1)$ random variables. Substituting $T_O=Z_1+b/\sigma_O$ and $Y_{R,GMM}=\sigma_U\sqrt{1-\rho^2}Z_2+\theta-b\rho\sigma_U/\sigma_O$ into (\ref{Bminimax_main_example_eq}) and rearranging terms yields
\begin{align}\label{sampling_distribution_eq}
\frac{\estimator(Y_U,T_O) - \theta}{\sigma_U}
=\rho \left[\bestimator\left(Z_1+\tilde b\right) - \tilde b\right] + \sqrt{1-\rho^2} Z_2,
\quad\text{where}\quad \tilde b=b/\sigma_O.
\end{align}

This representation holds under the distribution for $(Y_U,T_O)$ maintained in (\ref{main_example_eq}), which provides an asymptotic approximation under local misspecification.  In this asymptotic regime, 
consistent estimators of $\rho$, $\sigma_U$ and $\sigma_O$ are available via the usual asymptotic variance formulas used in overidentification tests for GMM.
In contrast, $b$ gives the limit of the bias of the restricted estimator divided by $\sqrt{n}$ and cannot be consistently estimated. Consequently, it is not possible to consistently estimate the asymptotic distribution of $\hat \theta(Y_U,T_O)$.

For example, the MSE of the estimator $\estimator(Y_U,T_O)$ is
\begin{align*}
\sigma^2_U \left[\rho^2 r(b/\sigma_O;\bestimator(\cdot)) + 1- \rho^2 \right],
\quad\text{where}\quad r(\tilde b;\bestimator(\cdot))= E_{T\sim N(\tilde b,1)} (\bestimator(T)-\tilde b)^2.
\end{align*}
Figures \ref{fig:build_oracle} and \ref{fig:pre_test} of Section \ref{sec:intro_example} plot this quantity as a function of $\tilde b$ with consistent estimates of $\rho$, $\sigma_U$, and $\sigma_O$ plugged in.  However, 
$\tilde b$ itself cannot be consistently estimated.
See \citet{leeb_model_2005} for a discussion of these issues in the context of pre-test estimators.

\subsubsection{Confidence Intervals}

Using (\ref{sampling_distribution_eq}), one can obtain a $100\cdot (1-\alpha)\%$ CI that is valid under the parameter space $\mathcal{C}_B=\mathbb{R}\times [-B,B]$ for $(\theta,b)$ by using a critical value $c_{\alpha}(\tilde B)=c_{\alpha}(\tilde B;\rho,\bestimator)$ solving
\begin{align}
\inf \chi \quad\text{s.t.}\quad \sup_{\tilde b:|\tilde b|\le \tilde B} P\left(\left|\rho \left[\bestimator\left(Z_1+\tilde b\right) - \tilde b\right] + \sqrt{1-\rho^2} Z_2\right| > \chi\right) \le \alpha. \label{eq: crit}
\end{align}
This critical value can then be used to form the \emph{fixed length confidence interval} (FLCI) $\left\{\hat\theta(Y_U,T_O) \pm \sigma_U c_{\alpha}(B/\sigma_O)\right\}$ centered at the estimator $\hat\theta(Y_U,T_O)$.
To emphasize the dependence on the parameter space $\mathcal{C}_B$ under which coverage is guaranteed, we will refer to such intervals as $B$-FLCIs.
For example, one can form the $B$-FLCI centered at the $B$-minimax estimator by using the critical value
$c_{\alpha}(B/\sigma_U)$ for this estimator.
Setting $B=\infty$, the $\infty$-FLCI centered at the $\infty$-minimax estimator is the usual CI centered at the unrestricted estimator: $\{Y_U\pm z_{1-\alpha/2}\sigma_U\}$.  This CI turns out to be larger than the $B$-FLCI centered at the $B$-minimax estimator for finite $B$, reflecting its validity over the larger parameter space $b\in \mathbb{R}$.

One can 
compute a $B$-FLCI centered at the adaptive estimator by computing the critical value $c_{\alpha}(B/\sigma_O;\rho,\tildedeltaadapt(\cdot;\rho^2))$ for the adaptive estimator. 
Unfortunately, it can be shown formally that 
any CI that is valid for all $b\in \mathbb{R}$ must have average length close to the length $2z_{1-\alpha/2}\sigma_U$ of the CI centered at $Y_U$, even if $b$ happens to be close to zero \citep[see Section 4 of][]{armstrong_sensitivity_2021}.
In light of this impossibility result, it is reasonable to report alongside an adaptive estimate the critical values for a $0$-FLCI and $\infty$-FLCI, thereby summarizing the range of critical values needed to guarantee coverage under different assumptions. When $|\rho|$ is large, the critical value for a $0$-FLCI will be far below the usual 1.96 benchmark for a 95\% test. Conversely, the corresponding critical value for a $\infty$-FLCI interval will be much larger than 1.96, reflecting the inherent tradeoffs involved in centering the CI around the adaptive estimator rather than the unbiased estimator. \citet{cai_adaptive_2005} discuss analogous tradeoffs involving centering in the context of nonparametric estimation.

An alternate approach, which we explore in our main empirical example, is to construct a $B$-FLCI for some intermediate value of $B$ and report both its worst and best case coverage. Researchers who are open to trading off some worst-case coverage for a shorter CI or enhanced best-case coverage might find an interval centered around an adaptive estimator, offering coverage (say) between 90\% and 97\%, more appealing than a longer interval centered around $Y_U$ that consistently provides 95\% coverage. This interval could also be preferable to a slightly shorter 90\% CI centered around $Y_U$, as the additional 7 percentage points of potential coverage may be more valuable than a modest reduction in length.

\subsection{Analytic adaptive estimators}\label{nearly_adaptive_computation_sec}

While the optimally adaptive estimator is 
trivial to implement once the solution is tabulated, it lacks a simple closed form. To reduce the opacity of the procedure, one can replace the term $\bestimator(T_O)$ in (\ref{Bminimax_main_example_eq}) with an analytic approximation. 
A natural choice of approximations for $\bestimator(T_O)$ is the class of \emph{soft-thresholding} estimators, which are indexed by a threshold $\lambda\ge 0$ and given by
\begin{align*}
\delta_{S,\lambda}(T)=\max\left\{|T|-\lambda,0\right\}\operatorname{sgn}(T)
=\begin{cases}
T-\lambda & \text{if } T>\lambda  \\
T+\lambda & \text{if } T<-\lambda  \\
0         & \text{if } |T|\le \lambda.
\end{cases}
\end{align*}
We also consider the class of \emph{hard-thresholding} estimators, which are given by
\begin{align*}
\delta_{H,\lambda}(T)=T\cdot I(|T|\ge \lambda)
=\begin{cases}
T         & \text{if } |T|>\lambda  \\
0         & \text{if } |T|\le \lambda.
\end{cases}
\end{align*}
Note that hard-thresholding leads to a simple pre-test rule: use the unrestricted estimator if $|T_O|>\lambda$ (i.e. if we reject the null that $b=0$ using critical value $\lambda$) and otherwise use the GMM estimator that is efficient under the restriction $b=0$.
The soft-thresholding estimator uses a similar idea, but avoids the discontinuity at $T_O=\lambda$.

A third estimator, which we will call the empirical risk minimizer (ERM), takes the form $\delta_{ERM}(T_{O})=\frac{T_{O}^{2}}{T_{O}^{2}+1}\cdot T_{O}$. The ERM estimator, which was proposed by \cite{de_chaisemartin_empirical_2020}, minimizes the estimated risk of the weighted average between $Y_{U}$ and $Y_{R,GMM}$.
The ERM can be generalized to a broader class of estimators 
$\delta_{ERM,\lambda}(T_{O})=\frac{T_{O}^{2}}{T_{O}^{2}+\lambda}\cdot T_{O}$,
which was briefly considered in \citet[][p. 230]{Magnus_Estimation_2002}.
We can optimize $\lambda$ for the worst-case adaptation regret given a specific value of $\rho^{2}$,
which yields the \emph{adaptive ERM} estimator.

To compute the adaptive ERM estimator along with the hard and soft-thresholding estimators that are optimally adaptive in these classes of estimators, we  numerically minimize (\ref{adaptation_objective_main_example_eq}) over $\lambda$  as explained in Online Appendix~\ref{appdx:st}.  We plot the respective optimal thresholds in Figure~\ref{fig:thresholds}, which are only a function of the relative efficiency $\sigma_{R,GMM}^2/\sigma_U^2=1-\rho^2$. We will be especially interested in the optimal soft-threshold, which can be closely approximated using the formula $\lambda= 0.45 - 0.24 \cdot \ln(1-\rho^2)$ for $\rho^2\in(0.002,0.99)$.%
\begin{figure}[h!]
\caption{\label{fig:thresholds}Thresholds minimizing the worst-case adaptation regret}

\begin{centering}
\includegraphics[scale=0.20]{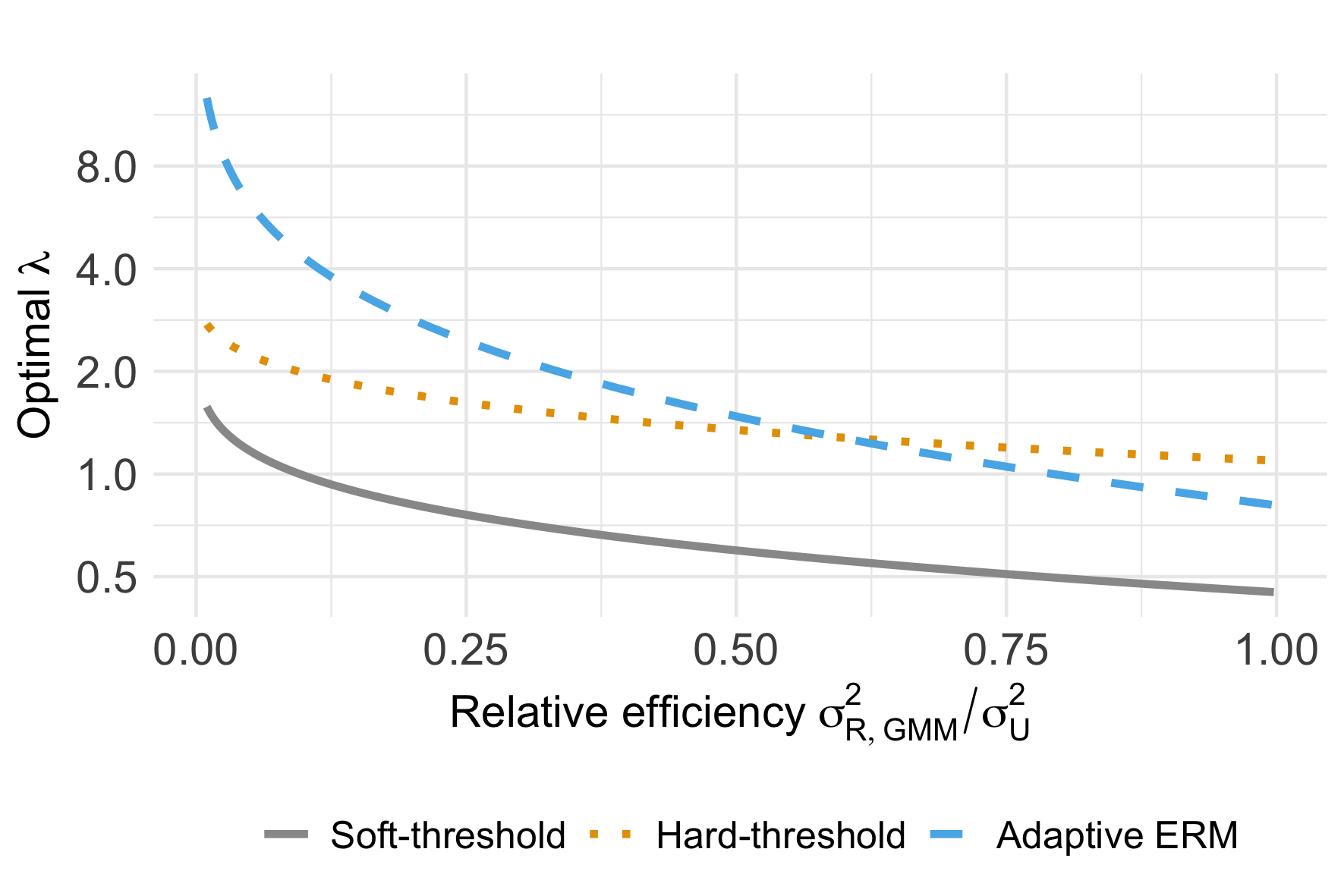}
\par\end{centering}
\centering{\footnotesize{Notes: Vertical axis is spaced on a $\log_2$ scale.}}

\end{figure}

\begin{figure}[h!]
\caption{\label{fig:policy} Estimators of scaled bias when $\sigma_{R,GMM}^2/\sigma_U^2=0.41$}
\begin{centering}
\includegraphics[scale=0.20, trim={0 1.5cm 0 0.5cm}, clip]{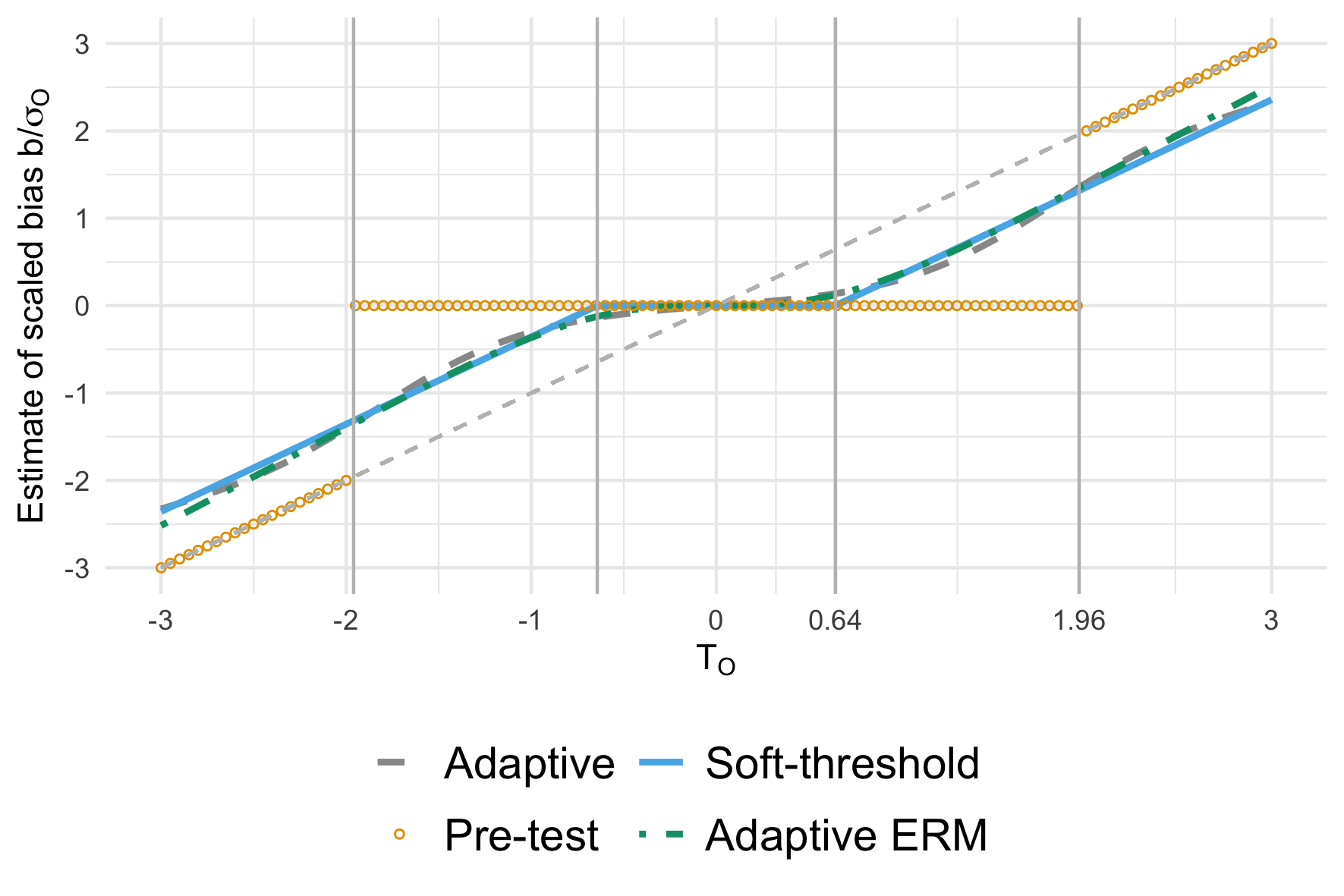}
\par\end{centering}
\footnotesize{Notes: Solid vertical line at 0.64 depicts optimal soft-threshold. Solid line at 1.96 depicts conventional pre-test threshold.}
\end{figure}

Figure \ref{fig:policy} plots the optimally adaptive and soft-thresholding estimators of the scaled bias as functions of $T_O$. To ease visual inspection of the differences between these estimators, they have been plotted over the restricted range [-3,3]. These functions depend on the data only through the estimated value of $1-\rho^2$, which takes the value 0.41 here, as in the two-way fixed effects example introduced in Section \ref{sec:intro_example}. The optimal soft-threshold $\lambda$ yielding the lowest worst cast adaptation regret in this example is 0.64. The optimally adaptive, adaptive ERM, and soft-thresholding estimators continuously shrink small values of $T_O$ towards zero. However, the soft-thresholding estimator sets all values of $|T_O|$ less than $0.64$ to zero, while the optimally adaptive and adaptive ERM estimators avoid flat regions. In contrast to the continuous nature of these adaptive estimators, a conventional pre-test using $\lambda=1.96$ exhibits large discontinuities at the hard-threshold.  The pre-test choice of $\lambda=1.96$ differs from the value that minimizes worst-case adaptation regret, which in this example is 1.43.

Like the optimally adaptive estimator $\deltaadapt$, the worst-case adaptation regret of the adaptive soft and hard-thresholding estimators depends only on $1-\rho^2$.  We report comparisons between these estimators in our empirical applications in Section~\ref{sec:Examples}.  As discussed in Online Appendix~\ref{appdx:nearly_adaptive_approx}, 
soft-thresholding yields nearly optimal performance for the adaptation problem relative to $\deltaadapt$ in a wide range of settings.
In contrast, hard-thresholding typically exhibits both substantially elevated worst case adaptation regret and worst case risk driven by the possibility that the scaled bias has magnitude near $\lambda$. The adaptive ERM estimator generally exhibits slightly higher worst case risk and adaptation regret than the soft-thresholding estimator but exhibits lower risk when the bias is very large.

Our finding that soft-thresholding is nearly optimal for adaptation mirrors the findings of \citet{bickel_parametric_1984} for the case where the set $\mathcal{B}$ of bounds $B$ on the bias consists of the two elements $0$ and $\infty$.
\citet[][p. 231]{Magnus_Estimation_2002} reports that soft-thresholding 
optimizes a related regret problem over a certain class of estimators indexed by two scalar parameters.
While soft-thresholding is perhaps the simplest way of achieving near-optimal performance for adaptation, other generalizations of thresholding estimators \citep[e.g.,][pp. 200-201]{johnstone_gaussian_2019} have been found to have similar risk properties to soft-thresholding, and may also perform well in our setting.

As detailed in Appendix~\ref{lasso_appendix}, the soft-thresholding estimator is numerically equivalent to a generalized lasso estimator \citep{tibshirani1996regression}
applied to a dataset comprised of the restricted and unrestricted estimates. The regressors are a constant and an indicator for the restricted estimate, the coefficient on which measures the bias $b$. The lasso penalty shrinks the bias estimate towards zero and depends only on the soft-threshold $\lambda$. Hence, the adaptive soft-threshold provides an optimal tuning of lasso for low-dimensional settings in which interest centers on a scalar parameter.
This exact tuning contrasts with high-dimensional settings where existing tuning methods typically only offer rate results. %

\subsection{Constrained adaptation}\label{constrained_adaptation_sec}

If the loss of efficiency under adaptation $A^*(\mathcal{B})$ is large, both the optimally adaptive estimator and its soft-thresholding approximation will possess worst case risk far above the oracle minimax risk, which limits their practical appeal.
As we show in Online Appendix~\ref{rho_to_one_section},
$A^*(\mathcal{B})$ will tend to be large when $|\rho|$ is large, which corresponds to settings where $Y_R$ is orders of magnitude more precise than $Y_U$. 

In such cases, it may be attractive to temper the degree of adaptation that takes place by restricting attention to estimators that exhibit worst case risk no greater than a constant $\bar R$. Online Appendix Section~\ref{constrained_adaptation_sec_appx} details how to compute such a constrained adaptive estimator. As noted by \citet{bickel_parametric_1984} in his analysis of the granular case where $\mathcal{B}=\{0,\infty\}$, it is often possible to greatly improve the risk at $b=0$ relative to the unbiased estimator $Y_U$ in exchange for modest increases in risk when $b=\infty$.  Similarly, we find that setting $\overline R$ to  $50\%$ above the risk of $Y_U$ yields large efficiency improvements when $b$ is small.

The constrained adaptive estimator bears some similarity to the ERM estimator. 
\cite{de_chaisemartin_empirical_2020} prove that the maximal risk decrease of $\delta_{ERM}$ relative to the risk of the unbiased estimator is larger than the maximal risk increase of $\delta_{ERM}$ relative to the unbiased estimator.  %
Through numerical calculations reported in a prior version of this paper \citep{armstrong2023adapting}, we find that this property holds for the constrained soft-thresholding version of our estimator so long as $\overline R$ is less than $70\%$ above the risk of $Y_U$. Remarkably, the property holds even for unconstrained soft-thresholding ($\overline R=\infty$) so long as $\rho^2$ is less than 0.86.

\section{Examples}\label{sec:Examples}

We now consider two empirical examples where questions of specification arise and examine how adapting to misspecification compares to pre-testing and other strategies such as committing ex-ante to either the unrestricted or restricted estimator. A third example, provided in Online \ref{sec:lalonde}, considers a multivariate adaptation problem with two restricted models and corresponding bias estimates. %

\subsection{Adapting to heterogeneous effects \citep{gentzkow_effect_2011}} \label{sec:gentzkow}

Returning to the example introduced in Section \ref{sec:intro_example}, Table \ref{tab: turnout} reports the realizations of $(Y_U,Y_R)$ and their standard errors, which exactly replicate those given in Table 3 of \citet{de_chaisemartin_two-way_2020} after dividing by 100. The estimated variance of $Y_O$ is closely approximated by the difference in squared standard errors between $Y_U$ and $Y_R$, suggesting $Y_R$ is nearly efficient. Hence, the downstream GMM, adaptive, and soft-thresholding estimators could have been accurately approximated using only the published point estimates and standard errors.  In contrast to the analysis in Section \ref{sec:intro_example}, we treat $Y_{R,GMM}$ rather than $Y_R$ as the efficient estimator, resulting in small differences from the previously reported downstream results. Standard errors are not reported for the soft-thresholding, adaptive, or pre-test estimators because the variability of these procedures depends on the unknown bias level $b$.

\begin{table}[h]
\begin{centering}
\caption{\label{tab: turnout} Estimates of the effect of an additional newspaper
on turnout.}
\begin{tabular}{cccccccccc}
\hline 
 &   &    &   & Pre- & Opt. & Soft- & Hard- &  &  Adapt\tabularnewline
 & $Y_{U}$  & $Y_{R}$  & $Y_{R,GMM}$ & test & Adapt & thresh & thresh &  ERM & ERM\tabularnewline
\hline 
\hline 
Estimate & 0.43 & 0.26  & 0.24 & 0.24 & 0.36 & 0.36 & 0.43 &  0.38 & 0.36 \tabularnewline
Std Error & (0.14) & (0.09)  & (0.09) &  &  &  &  &   \tabularnewline
\hline
Max Risk & 0\% & $\infty$ & $\infty$ & 87\% & 39\% & 25\% & 39\% & 15\% & 25\%\tabularnewline
Max Regret & 145\% & $\infty$ & $\infty$ & 134\% & 44\% & 46\% & 82\% & 68\% & 50\%\tabularnewline
Threshold &  &  & & 1.96 &  & 0.64  & 1.43 & 1  & 1.73\tabularnewline
\hline \hline
\end{tabular}
\par\end{centering}
\footnotesize{Notes: Bootstrap standard
errors in parentheses computed using the same 100 bootstrap samples utilized by \citet{de_chaisemartin_two-way_2020}. The over-identification test statistic is $T_O=-1.75$. ``Pre-test'' selects between $Y_U$ and GMM based on $|T_O|\gtrless1.96\sigma_O$. The relative efficiency of $Y_U$ to $Y_{R,GMM}$ is $1-\rho^2=0.41$. ``Max Risk'' gives the percentage increase in worst case risk over $Y_U$: $(\sup_B R_{\max}(B,\estimator)/\sigma_U^2-1)\times 100$. ``Max Regret'' refers to the worst case adaptation regret in percentage terms $(A_{\max}(\mathcal{B},\estimator)-1)\times100$.  %
 }
\end{table}

Both $Y_R$ and $Y_{R,GMM}$ exhibit standard errors roughly 35\% below that of $Y_U$. Consequently, relying solely on the convex-weighted estimator $Y_U$ exposes the researcher to a large worst-case adaptation regret of 145\%. Though the realized value of $Y_U$ is nearly twice as large as that of $Y_{R,GMM}$, the two estimators are not statistically distinguishable from one another at the 5\% level. Hence, a conventional pre-test suggests ignoring the perils of negative weights and confining attention to $Y_{R,GMM}$ on account of its substantially increased precision. The worst case MSE of the pre-test estimator, which exhibits a hump shaped risk profile similar to that depicted in Figure \ref{fig:pre_test}, is 87\% higher than the MSE $\sigma_U^2$ of $Y_U$. Pre-testing also yields sizable worst-case adaptation regret reflecting the possibility that the test selects the inefficient $Y_U$ when $b=0$. 

In contrast to the pre-test estimator, both the optimally adaptive estimator and its soft-thresholding approximation place substantial weight $w(T_O)$ on the convex estimator, yielding estimates roughly 60\% of the way towards $Y_U$ from $Y_{R,GMM}$. This phenomenon owes to the fact that with $T_O=-1.75$ both estimators detect the presence of a non-trivial amount of bias in $Y_R$. We can easily compute the soft-thresholding bias estimate from the figures reported in the table as $(-1.75+.64)\times-0.77\times0.14\approx 0.12$, suggesting that $Y_{R,GMM}$ exhibits a bias of roughly 50\%. Balancing this bias against the estimator's increased precision leads the soft-thresholding estimator to essentially split the difference between the convex and non-convex weighted estimators. 

By construction, the adaptive estimator exhibits lower worst case adaptation regret than the soft-thresholding estimator but the differences are quantitatively trivial. However, the soft-thresholding estimator exhibits meaningfully lower worst case risk than the adaptive estimator. Though the two estimators happen to yield identical estimates ex-post in this example, the ex-ante risk properties of the adaptive soft-thresholding estimator arguably commend it over the optimally adaptive estimator. 

The ERM estimator of \cite{de_chaisemartin_empirical_2020} yields lower worst case risk than soft-thresholding but substantially larger adaptation regret. Optimizing the ERM threshold to minimize adaptation regret yields worst case risk equivalent to the soft-thresholding estimator but higher adaptation regret. Of the estimators considered, soft-thresholding offers the most attractive tradeoff between worst case risk and adaptation regret.

\textbf{Confidence Intervals}
Table \ref{tab:coverage} reports the best case and worst case coverage of a series of confidence intervals. The first two columns of Panel A show that the usual 95\% confidence interval centered around the unbiased estimator has proper size, while a naive CI centered around the restricted estimator has best case coverage of 95\% and worst case coverage of 0\% attributable to the potentially unlimited bias of the restricted estimator. Relying on a pre-test to select one of these two confidence intervals yields a minimum coverage level of 67\%. By contrast, centering a CI around the optimally adaptive estimator using the standard error of the unbiased estimator yields best case coverage of 98\% and worst case coverage of 90\%. Centering around the soft-thresholding estimator yields even more favorable results, raising the worst case coverage to 93\%.

\begin{table}[h!]
\caption{\label{tab:coverage} Coverage and length of confidence intervals} 
\begin{centering}

\textbf{Panel A: Simple CIs} \\[0.2cm]
\begin{tabular}{cccccc}
\hline
 &  &  & & Opt. & Soft- \tabularnewline

 & $Y_{U}$ & $Y_{R}$ & Pre- & Adapt & Thresh \tabularnewline
 & $\pm 1.96\sigma_U$  & $\pm 1.96\sigma_R$ & test & $\pm 1.96\sigma_U$ & $\pm 1.96\sigma_U$ \tabularnewline
\hline
\hline
Max Coverage & 95\% & 95\% & 95\% & 98\% & 98\% \tabularnewline
Min Coverage & 95\% & 0\% & 67\% & 90\% & 92\% \tabularnewline
\hline \hline
\end{tabular}

\vspace{0.5cm} %

\textbf{Panel B: $B$-FLCIs} \\[0.2cm]
\begin{tabular}{cccccccc}
\hline
 & Opt. &Soft-  & Opt. & Soft- & Opt. & Soft-
\tabularnewline
 & Adapt & Thresh & Adapt & Thresh & Adapt & Thresh \tabularnewline
 & $\pm c_{.05}(0) \sigma_U$ & $\pm c_{.05}(0) \sigma_U$ & $\pm c_{.05}(1) \sigma_U$ & $\pm c_{.05}(1) \sigma_U$ & $\pm c_{.05}(9) \sigma_U$ & $\pm c_{.05}(9) \sigma_U$ \tabularnewline
\hline
\hline
Max Coverage & 95\% & 95\% & 97\% & 97\% & 99\% & 99\% \tabularnewline
Min Coverage & 80\% & 87\% & 86\% & 90\% & 95\% & 95\% \tabularnewline
Critical Val & 1.54 & 1.62 & 1.74 & 1.77 & 2.32 & 2.11 \tabularnewline
\hline \hline
\end{tabular}

\par\end{centering}
\footnotesize{Notes: ``Max coverage'' refers to the maximal coverage probability for the given confidence interval. ``Min Coverage'' refers to the min coverage probability. ``Adaptive'' refers to the optimally adaptive estimator and ``Soft-Thresh'' refers to soft-thresholding. ``Pre-test'' switches between $Y_U\pm1.96\sigma_U$ and $Y_R\pm1.96\sigma_R$ based on whether $|T_O|\gtrless 1.96 \sigma_O$. Critical values for $B$-FLCIs found by solving \eqref{eq: crit}. Min/max coverage evaluated using the expression for the constraint in \eqref{eq: crit}.}
\end{table}

Panel B of Table \ref{tab:coverage} considers $B$-FLCIs centered around the adaptive estimators. A $0$-FLCI centered around the optimally adaptive estimator has a half length of only about $1.54\sigma_U$ (as opposed to the traditional $1.96\sigma_U$ utilized in Panel A) but exhibits worst case coverage of 80\%. Centering around the soft-thresholding estimator yields a slightly longer interval, which improves minimum coverage to 87\%. The third row of Panel B shows the coverage of a $\sigma_O$-FLCI centered around the optimally adaptive estimator, which yields modestly longer CI but lowers worst case coverage to 86\%. Again, centering at the soft-thresholding estimator raises worst case coverage slightly, in this case to 90\%. Finally, we approximate an $\infty$-FLCI by setting $B=9\sigma_O$, which yields very conservative intervals with half-lengths exceeding $2.1\sigma_U$. 

The simplicity and robustness of intervals based upon an adaptive estimator $\pm 1.96\sigma_U$ make them an attractive option. For researchers who seek shorter intervals, the $\sigma_O$-FLCI centered around the soft-thresholding estimator seems to offer a reasonable mix of worst and best case coverage. Notably, all of these options offer substantially higher worst case coverage than pre-testing, which remains widespread in applied research.

\subsection{Adapting to endogeneity \citep{angrist_does_1991}}
Our second example, which is meant to highlight the limits of optimal adaptation, comes from \citet{angrist_does_1991}'s seminal analysis of the returns to schooling using quarter of birth as an instrument for schooling attainment. 
Table~\ref{tab:AK91_estimates} replicates exactly the estimates reported in \citet[][Panel B, Table III]{angrist_does_1991} for men born 1930-39. $Y_U$ gives the Wald-IV estimate of the returns to schooling using an indicator for being born in the first quarter of the year as an instrument for years of schooling completed, while $Y_R$ gives the corresponding OLS estimate. Neither estimator controls for additional covariates.

The first stage relationship between quarter of birth and years of schooling exhibits a z-score of 8.22, suggesting an asymptotic normal approximation to $Y_U$ is likely to be highly accurate. We follow the original study in assuming homoscedasticity, in which case OLS ($Y_R$) is known to be the asymptotically efficient GMM estimator under exogeneity.

\begin{table}[h!]
\caption{\label{tab:AK91_estimates} Estimates of the return to an additional year of schooling.}
\begin{centering}
\begin{tabular}{ccccccccc}
\hline 
&  &  &  &  \multicolumn{3}{c}{\underline{Unconstrained}} & \multicolumn{2}{c}{ \underline{Constrained}}\tabularnewline 
&  &  &Pre-   & Opt. & Soft- & Hard- & Opt. & Soft-\tabularnewline
 & $Y_{U}$ & $Y_{R}$ & test & Adapt & thresh & thresh & Adapt & thresh\tabularnewline
\hline 
\hline 
Estimate & 0.102 & 0.071 & 0.071 & 0.071 & 0.071 & 0.071 & 0.080 & 0.085\tabularnewline
Std Error & (0.0239) & (0.0003) &  &  &  & &  & \tabularnewline \hline
Max Risk & 0\% & $\infty$ & 147\%  & 465\% &  440\% &  521\% & 50\% & 50\%\tabularnewline
Max Regret & 500,145\% & $\infty$ & 21,081\% &505\% &  552\% &  724\%& 17,375\% &  20,579\%\tabularnewline
Threshold &  &  & 1.96 &  & 2.10& 3.34 &  & 0.71\tabularnewline
\hline \hline
\end{tabular}
\par\end{centering}

\footnotesize{Notes: Standard
errors in parentheses computed under homoscedasticity as in original study. Under homoscedasticity, $Y_R$ coincides with GMM. The over-identification test statistic is $T_O=-1.3$. ``Max Risk'' gives the percentage increase in worst case risk over $Y_U$: $(\sup_B R_{\max}(B,\estimator)/\sigma_U^2-1)\times 100$. ``Max regret'' refers to the worst case adaptation regret in percentage terms $(A_{\max}(\mathcal{B},\estimator)-1)\times100$. The relative efficiency of $Y_U$ to $Y_R=Y_{R,GMM}$ is $1-\rho^2 = 0.0002$.%
}
\end{table}

While the IV estimator accounts for endogeneity, it is highly imprecise, with a standard error two orders of magnitude greater than OLS. Consequently, the maximal regret associated with using IV instead of OLS is extremely large, as $Y_U$ is only 0.02\% as efficient as $Y_R$ when exogeneity holds. IV and OLS cannot be statistically distinguished at conventional significance levels, with $T_O\approx-1.3$. The inability to distinguish IV from OLS estimates of the returns to schooling is characteristic not only of the specifications reported in \citet{angrist_does_1991} but of the broader quasi-experimental literature spawned by their landmark study \citep{card1999causal}.

The confluence of extremely large maximal regret for $Y_U$ with a statistically insignificant difference $Y_O$, leads the adaptive estimator, the soft-thresholding estimator, and the pre-test estimator to all coincide with $Y_R$. 
Despite the agreement of the three approaches, the extremely large adaptation regret exhibited by the optimally adaptive estimator suggests it is unlikely to garner consensus in this setting. 
While the adaptive and soft-thresholding estimators
avoid committing to either $Y_U$ or $Y_R$ before observing the data, they
still expose the researcher to more than a 400\% increase in maximal risk over $Y_U$. A skeptic concerned with the potential biases in OLS is therefore unlikely to be willing to rely on such an estimator.

If we instead limit ourselves to a 50\% increase in maximal risk, the adaptive and soft-threshold estimators yield returns to schooling estimates of 0.080 and 0.085 respectively. While the former estimate is a bit closer to OLS than IV, the latter is approximately halfway between the two. The maximal regret of both these estimators is extremely high, reflecting the potential efficiency costs of weighting $Y_U$ so heavily.
These efficiency concerns are likely outweighed in this case by the potential for extremely large biases.

\section{Conclusion}\label{conclusion_sec}

Empiricists routinely encounter robustness-efficiency tradeoffs. The reporting of estimates from different models has emerged as a best practice at leading journals. The methods introduced here provide a scientific means of summarizing what has been learned from such exercises and arriving at a preferred estimate that trades off considerations of bias against variance. 

Computing the adaptive estimators proposed in this paper requires only point estimates, standard errors, and the covariance between estimators, objects that are easily produced by standard statistical packages. As our examples revealed, in many cases the restricted estimator is nearly efficient, implying the relevant covariance can be deduced from the standard errors of the restricted and unrestricted estimators. 

In line with earlier results from \cite{bickel_parametric_1984}, we found that soft-thresholding estimators closely approximate the optimally adaptive estimator in the scalar case, while requiring less effort to compute. An interesting topic for future research is whether similar approximations can be developed for higher dimensional settings where the curse of dimensionality renders direct computation of optimally adaptive estimators infeasible.

\bibliographystyle{chicago}
\bibliography{references_low_dim_adapt,references_empirical,additional_references}

\appendix
\renewcommand{\thesection}{Appendix \Alph{section}}
\renewcommand{\thesubsection}{\Alph{section}.\arabic{subsection}}
\setcounter{theorem}{0}
\renewcommand{\thetheorem}{\Alph{section}.\arabic{theorem}}
\setcounter{lemma}{0}
\renewcommand{\thelemma}{\Alph{section}.\arabic{lemma}}
\setcounter{table}{0}
\renewcommand{\thetable}{A\arabic{table}}
\setcounter{figure}{0}
\renewcommand{\thefigure}{A\arabic{figure}}

\section{Group decision making interpretation}\label{appdx:group}

This appendix develops a stylized model of group decision making inspired by \citeOnline{savage_foundations_1954}'s arguments regarding the ability of minimax decisions to foster consensus among individuals with heterogeneous beliefs. Extending these arguments, we illustrate how adaptive decisions can serve to foster consensus across groups of individuals with different sets of beliefs.

\subsection{Consensus in a single committee}
Suppose there is a committee charged with deciding on the value of a parameter $\theta$ based on the evidence $(Y_U,Y_R)$. The committee is comprised of members with heterogeneous beliefs over $(\theta,b)$ that include all priors supported on the set $\mathcal{C}_{B}$. The committee chair, who we will call the \emph{$B$-chair}, offers a take it or leave it proposal that her committee agree on the estimator $\estimator$ in exchange for the provision of a public good providing payoff $G$ to each member of the committee. 

If the committee agrees to the proposal, the $B$-chair earns a payoff $K-C(G)$, where $K$ is the value of consensus and $C(\cdot)$ is an increasing cost function. If some member of the committee does not agree to the proposal, the chair and all committee members receive payoff zero. The $B$-chair therefore seeks an estimator $\estimator$ allowing payment of the smallest $G$ that ensures consensus.

A committee member who is certain of the parameters $\left(\theta,b\right)$ will accept the chair's offer if and only if $R\left(\theta,b,\estimator\right)\leq G$. However, the committee member with the most pessimistic beliefs %
will require a public goods provision level of at least $R_{\max}\left(B,\estimator\right)$ to agree to the offer.  To achieve consensus at minimal cost, the $B$-chair can propose the $B$-minimax estimator, which requires public goods provision level $R^{*}\left(B\right)$ to achieve consensus. 

The $B$-chair will be willing to provide this level of public goods if and only if $K \geq C(R^{*}\left(B\right))$, in which case consensus ensues. If this condition does not hold, the chair deems the $B$-minimax estimator too costly to implement and consensus is not achieved. 

\subsection{Consensus among committees}

Now suppose there is a collection $\mathcal{B}$ of committees,
each of which must decide on the parameter $\theta$ using $(Y_U,Y_R)$. This collection is led by a \emph{chair of chairs} (CoC) who would like for the $B$-chairs to agree on a common estimator $\estimator$. Suppose also that $K>\sup_{B\in\mathcal{B}}C(R^{*}\left(B\right))$, so that each $B$-chair would privately prefer the $B$-minimax estimator. The CoC has a fixed budget $F>0$ that can be used to provide a public good $\tilde G$ enjoyed by all chairs. The CoC makes provision of $\tilde G$ contingent on the agreement of all $B$-chairs to use $\estimator$: if they fail to reach consensus, the public good is not provided. The cost to the CoC of providing public goods level $\tilde G$ is $\tilde C (\tilde G)$, where $\tilde C(\cdot)$ is monotone increasing.

By the arguments above, each $B$-chair must pay a cost $C( R_{\max}\left(B,\estimator\right))$ to secure consensus regarding the CoC's proposed $\estimator$, leaving her with payoff $K-C( R_{\max}\left(B,\estimator\right))$. However, each chair can also defy the CoC and propose the $B$-minimax estimator to her committee, yielding payoff $K- C(R^{*}\left(B\right))$. Hence, to compel a $B$-chair to use $\estimator$, the CoC must offer a public good providing utility of at least $\Delta_B(\estimator) = C( R_{\max}\left(B,\estimator\right))-C(R^{*}\left(B\right))$. To minimize costs, the CoC sets $\tilde G=\sup_{B\in\mathcal{B}} \Delta_B(\estimator)$, which is the level required to appease the most reticent $B$-chair. 

Different functional forms for the cost function $C$ yield different notions of adaptation. To motivate the formulation in \eqref{Aopt_def_eq}, we assume $C(G)\propto\ln G$, which implies chairs produce the public good using a technology that is exponential in costs. With this choice of $C(\cdot)$, the CoC's problem is to find a $\estimator$ that minimizes $\sup_{B\in \mathcal{B}} \ln \left(R_{\max}\left(B,\estimator\right) / R^{*}\left(B\right)\right)=\sup_{B\in \mathcal{B}} \ln A(B,\estimator)$.
The CoC will therefore propose the optimally adaptive estimator $\estimator^{*}$, which yields $\sup_{B\in\mathcal{B}} \Delta_B(\estimator^*)\propto\ln \Aopt(\mathcal{B})$. When $\tilde C(\ln \Aopt(\mathcal{B}))>F$, the CoC balks at the cost of implementing $\estimator^*$ and consensus fails.

\subsection{Discussion}
The prospects for achieving consensus are governed by the loss of efficiency under adaptation. When $\Aopt(\mathcal{B})$ is small, consensus is likely, as the adaptive estimator will yield maximal risk similar to each committee's perceived $B$-minimax risk. When $\Aopt(\mathcal{B})$ is large, however, consensus is unlikely to emerge, as the optimally adaptive estimator will be perceived as excessively risky by committees with extreme beliefs. 

\section{Details and proofs} %

\subsection{Details for Theorem \ref{main_example_thm_main_text} and extensions}\label{details_and_proofs_appendix}

Consider a slight extension of \eqref{main_example_eq} with $p$ misspecified estimates, leading to a $p\times 1$ vector $Y_O$:
\begin{align}\label{Y_U_Y_O_multivariate_def}
  Y=\left(
  \begin{array}{c}
    \underset{1\times 1}{Y_U}  \\
    \underset{p\times 1}{Y_O}
  \end{array}
  \right)
  \sim N\left(
  \left( 
  \begin{array}{c}
    \underset{1\times 1}{\theta}  \\
    \underset{p\times 1}{b}
  \end{array}
  \right),
  \Sigma
  \right),
  \quad
  \Sigma
  =\left(
  \begin{array}{cc}
    \underset{1\times 1}{\Sigma_U} & \underset{1\times p}{\Sigma_{UO}}  \\
    \underset{p\times 1}{\Sigma_{UO}'} & \underset{p\times p}{\Sigma_O}
  \end{array}
  \right).
\end{align}
In our main example, $p=1$ and $\rho=\Sigma_{UO}/\sqrt{\Sigma_U\Sigma_O}$.
We are interested in the minimax risk of an estimator $\delta:\mathbb{R}^{p+1}\to
\mathbb{R}$ under the loss function $L(\theta,b,d)$, which may incorporate a scaling to turn the minimax problem into a problem of finding an optimally adaptive estimator, following Lemma \ref{adaptation_as_weighted_minimax_lemma}.
The loss function satisfies the invariance condition
\begin{align}
  \label{eq:loss_invariance}
  L(\theta+t,b,d+t)=L(\theta,b,d)\quad\text{all }t\in\mathbb{R}.
\end{align}
We consider minimax estimation over a parameter space $\mathbb{R}\times
\mathcal{C}$:
\begin{align}\label{eq:minimax_problem_main_example_general}
  \inf_{\delta}\sup_{\theta\in\mathbb{R}, b\in\mathcal{C}} R(\theta,b,\delta).
\end{align}

\begin{theorem}\label{general_misspecification_theorem}
  Suppose that the loss function $L(\theta,b,d)$ is convex in $d$ and that
  (\ref{eq:loss_invariance}) holds.
  Then the minimax risk (\ref{eq:minimax_problem_main_example_general}) is
  given by
  \begin{align}\label{eq:general_minimax_b}
    &\inf_{\bar\delta}\sup_{b\in\mathcal{C}} E_{0,b}[\tilde L(b,\bar \delta(Y_O)-\Sigma_{UO}\Sigma_O^{-1}b)]  \\
    &=\sup_{\pi \text{ supported on }\mathcal{C}}\inf_{\bar\delta} \int E_{0,b}[\tilde L(b,\bar \delta(Y_O)-\Sigma_{UO}\Sigma_O^{-1}b)]\, d\pi(b)  \nonumber
  \end{align}
  where
  $\tilde L(b,t)=E L(0,b,t+V)$
  with $V\sim N(0,\Sigma_U-\Sigma_{UO}\Sigma_O^{-1}\Sigma_{UO}')$.
  Furthermore, the minimax problem (\ref{eq:minimax_problem_main_example_general}) has at least one
  solution, and any solution $\delta^*$ takes the form
  \begin{align*}
    \delta^*(Y_U,Y_O)=Y_U-\Sigma_{UO}\Sigma_{O}^{-1}Y_O + \bar \delta^*(Y_O)
  \end{align*}
  where $\bar \delta^*$ achieves the infimum in (\ref{eq:general_minimax_b}).
\end{theorem}
\begin{proof}
  The minimax problem (\ref{eq:minimax_problem_main_example_general}) is
  invariant (in the sense of pp. 159-161 of \citet{lehmann_theory_1998})
  to the transformations $(\theta,b)\mapsto (\theta+t,b)$ and the
  associated transformation of the data $(Y_U,Y_O)\mapsto (Y_U+t,Y_O)$, where $t$
  varies over $\mathbb{R}$.  Equivariant estimators for this group of
  transformations are those that satisfy $\delta(y_U+t,y_O)=\delta(y_U,y_O)+t$,
  which is equivalent to imposing that the estimator takes the form
  $\delta(y_U,y_O)=\delta(0,y_O)+y_U$.
  The risk of such an estimator does not depend on $\theta$ and is given by
  \begin{align*}
    R(\theta,b,\delta)=R(0,b,\delta)
    =E_{0,b}\left[ L(0,b,\delta(0,Y_O)+Y_U) \right].
  \end{align*}
  Using the decomposition $Y_U-\theta=\Sigma_{UO}\Sigma^{-1}(Y_O-b)+V$
  where
  $V\sim N(0,\Sigma_U-\Sigma_{UO}\Sigma_O^{-1}\Sigma_{UO}')$ is independent of
  $Y_O$, the above display is equal to
  \begin{align*}
    E_{0,b}\left[ L(0,b,\delta(0,Y_O)+\Sigma_{UO}\Sigma_O^{-1}(Y_O-b)+V) \right]
    =E_{0,b} \tilde L(b,\delta(0,Y_O)+\Sigma_{UO}\Sigma_O^{-1}(Y_O-b)).
  \end{align*}
  Letting $\bar\delta(Y_O)=\delta(0,Y_O)+\Sigma_{UO}\Sigma_O^{-1}Y_O$, the above
  display is equal to $E_{0,b}[\tilde L(b,\bar
  \delta(Y_O)-\Sigma_{UO}\Sigma_O^{-1}b)]$.  Thus, if an estimator $\bar
  \delta^*$ achieves the infimum in (\ref{eq:general_minimax_b}), the
  corresponding estimator
  $\delta(Y_U,Y_O)=\delta(0,Y_O)+Y_U=\bar\delta^*(Y_O)-\Sigma_{UO}\Sigma_O^{-1}Y_O+Y_U$
  will be minimax among equivariant estimators for
  (\ref{eq:minimax_problem_main_example_general}).
  It will then follow from the Hunt-Stein Theorem \citep[][Theorem
  9.2]{lehmann_theory_1998} that this minimax equivariant estimator is minimax
  among all estimators, that any other minimax estimator takes this form and that the minimax risk is given by the first line of
  (\ref{eq:general_minimax_b}).

  It remains to show that the infimum in the first line of
  (\ref{eq:general_minimax_b}) is achieved, and that the equality claimed in
  (\ref{eq:general_minimax_b}) holds.  The equality in
  (\ref{eq:general_minimax_b}) follows from the minimax theorem, as stated in
  Theorem A.5 in \citet{johnstone_gaussian_2019} (note that $d\mapsto \tilde
  L(b,d-\Sigma_{UO}\Sigma_O^{-1}b)$ is convex since it is an integral of the
  convex functions $d\mapsto L(0,b,d-\Sigma_{UO}\Sigma_O^{-1}b+v)$ over the
  index $v$).
  The existence of an estimator $\bar \delta^*$ that achieves the infimum
  in the first line of (\ref{eq:general_minimax_b}) follows by noting that the
  set of decision rules (allowing for randomized decision rules) is compact in the topology defined on p. 405
  of \citet{johnstone_gaussian_2019}, and the risk
  $E_{0,b}[\tilde L(b,\bar\delta(Y_O)-\Sigma_{UO}\Sigma_{O}^{-1}b)]$ is continuous in
  $\bar\delta$ under this topology.  As noted immediately after Theorem A.1 in
  \citet{johnstone_gaussian_2019}, this implies that
  $\bar\delta\mapsto \sup_{b} E_{0,b}[\tilde
  L(b,\bar\delta(Y_O)-\Sigma_{UO}\Sigma_{O}^{-1}b)]$ is a lower semicontinuous function on
  the compact set of possibly randomized decision rules under this topology, which means
  that there exists a decision rule that achieves the minimum.  From this
  possibly randomized decision rule, we can construct a nonrandomized decision
  rule that achieves the minimum by constructing a nonrandomized decision rule
  with uniformly smaller risk by averaging, following
  \citet[][p. 404]{johnstone_gaussian_2019}.
\end{proof}

We now prove Theorem \ref{main_example_thm_main_text} by specializing this result.  
Note that $\Sigma_U$ and $\Sigma_O$ correspond to $\sigma_U^2$ and $\sigma_O^2$ in the main text respectively, and that $\rho$ in the main text is given by $\Sigma_{UO}/\sqrt{\Sigma_{U}\Sigma_{O}}$.
First, we derive the minimax estimator and minimax risk in
(\ref{eq:minimax_problem_main_example_general}) when
$L(\theta,b,d)=(\theta-d)^2$ and $\mathcal{C}=[-B,B]$.  We have
$\tilde L(b,t)=E(t+V)^2=t^2+\Sigma_{U}-\Sigma_{UO}^2/\Sigma_{O}$.  Thus,
(\ref{eq:general_minimax_b}) becomes
\begin{align*}
  &\inf_{\bar\bestimator}\sup_{b\in[-B,B]} E_{0,b}\left[\left(\bar\bestimator(Y_O)-\frac{\Sigma_{UO}}{\Sigma_{O}}b\right)^2\right] + \Sigma_{U}-\frac{\Sigma_{UO}^2}{\Sigma_{O}}  \\
  &=\inf_{\bar\bestimator}\sup_{b\in[-B,B]} \frac{\Sigma_{UO}^2}{\Sigma_O}E_{0,b}\left[\left(\frac{\sqrt{\Sigma_O}}{\Sigma_{UO}}\bar \delta(Y_O)-\frac{b}{\sqrt{\Sigma_{O}}}\right)^2\right] + \Sigma_{U}-\frac{\Sigma_{UO}^2}{\Sigma_{O}}.
\end{align*}
This is equivalent to observing $T_O=Y_O/\sqrt{\Sigma_O}\sim N(t,1)$ and
finding the minimax estimator of $t$ under the constraint
$|t|\le B/\sqrt{\Sigma_O}$.  Letting $\deltaBNM(T_O;B)$ denote
the solution to this minimax problem and letting $\rBNM(B/\sqrt{\Sigma_O})$
denote the value of this minimax problem, the optimal $\bar \delta$ in the above
display satisfies $\frac{\sqrt{\Sigma_O}}{\Sigma_{UO}}\bar
\delta(Y_O)=\deltaBNM(Y_O/\sqrt{\Sigma_O};B)$, which gives the
value of the above display as
\begin{align}\label{eq:B_minimax_risk_univariate_appendix}
  \frac{\Sigma_{UO}^2}{\Sigma_O}\rBNM(B/\sqrt{\Sigma_O}) + \Sigma_{U}-\frac{\Sigma_{UO}^2}{\Sigma_{O}}
\end{align}
and the $B$-minimax estimator as
\begin{align}\label{eq:B_minimax_estimator_univariate_appendix}
  \frac{\Sigma_{UO}}{\sqrt{\Sigma_O}}\deltaBNM(Y_O/\sqrt{\Sigma_O};B)
  +Y_U-\frac{\Sigma_{UO}}{\Sigma_{O}}Y_O.
\end{align}
Substituting $T_O=Y_O/\sqrt{\Sigma_O}$ and the notation $\rho=\Sigma_{UO}/\sqrt{\Sigma_U\Sigma_O}$, $\sigma_U^2=\Sigma_U$ and $\sigma_O^2=\Sigma_O$ used in the
main text gives (\ref{Bminimax_main_example_eq}) and (\ref{Rminimax_main_example_eq}).
This proves part (\ref{main_example_thm_minimax_estimator}) of Theorem \ref{main_example_thm_main_text}.

To find the optimally adaptive estimator and loss of efficiency under adaptation
in our main example, we apply Lemma \ref{adaptation_as_weighted_minimax_lemma} with $\omega(\theta,b)=\Rminimax(|b|)^{-1}$, with
$\Rminimax(B)$ given by (\ref{eq:B_minimax_risk_univariate_appendix}).  This
leads to the minimax problem (\ref{eq:minimax_problem_main_example_general})
with $\mathcal{C}=\mathbb{R}$ and
$L(\theta,b,d)=\Rminimax(|b|)^{-1}(\theta-d)^2$.  The function $\tilde L$
in Theorem \ref{general_misspecification_theorem} is then given by
$\tilde L(b,t)=E\Rminimax(|b|)^{-1}(t+V)^2
=\Rminimax(|b|)^{-1}(t^2+\Sigma_U-\Sigma_{UO}^2/\Sigma_O)$,
which gives (\ref{eq:general_minimax_b}) as
\begin{align*}
  &\inf_{\bar \delta}\sup_{b\in\mathbb{R}} \frac{E_{0,b}\left[\left(\bar\delta(Y_O)-\frac{\Sigma_{UO}}{\Sigma_O}b\right)^2\right]+\Sigma_U-\frac{\Sigma_{UO}^2}{\Sigma_O}}{\frac{\Sigma_{UO}^2}{\Sigma_O}\rBNM(|b|/\sqrt{\Sigma_O}) + \Sigma_{U}-\frac{\Sigma_{UO}^2}{\Sigma_{O}}}  %
  =\inf_{\bar \delta}\sup_{b\in\mathbb{R}} \frac{E_{0,b}\left[\left(\frac{\sqrt{\Sigma_O}}{\Sigma_{UO}}\bar\delta(Y_O)-\frac{b}{\sqrt{\Sigma_O}}\right)^2\right]+\rho^{-2}-1}{\rBNM(|b|/\sqrt{\Sigma_O}) + \rho^{-2}-1}.
\end{align*}
This proves part (\ref{main_example_thm_lea}) of Theorem \ref{main_example_thm_main_text}.
The above display is minimized by $\bar\delta$ satisfying
$\frac{\sqrt{\Sigma_O}}{\Sigma_{UO}}\bar\delta(Y_O)=\tildedeltaadapt(Y_O/\sqrt{\Sigma_O};\rho^2)$
where $\tildedeltaadapt(T;\rho^2)$ 
minimizes (\ref{adaptation_objective_main_example_eq}) in the main text.
By Theorem \ref{general_misspecification_theorem}, the optimally adaptive
estimator is given by
\begin{align}\label{eq:adaptive_estimator_main_example_appendix}
  \frac{\Sigma_{UO}}{\sqrt{\Sigma_O}}\tildedeltaadapt(Y_O/\sqrt{\Sigma};\rho^2)
  +Y_U-\frac{\Sigma_{UO}}{\Sigma_O}Y_O
  =\rho\sqrt{\Sigma_U}\tildedeltaadapt(T_O;\rho^2)
  +Y_U- \rho\sqrt{\Sigma_U}T_O.
\end{align}
This proves the part (\ref{main_example_thm_adaptive_estimator}) of Theorem \ref{main_example_thm_main_text}.

\subsection{Lasso interpretation of soft-thresholding}\label{lasso_appendix}

To connect the soft-thresholding estimator to lasso, consider a dataset with two observations comprised of the realizations of $Y_U$ and $Y_R$, and a linear model relating these estimates to a constant and an indicator for whether the observation is from the restricted specification. Letting $y_1=Y_U$, $d_1=0$, $y_2=Y_R$, and $d_2=1$, the model can be written
\begin{align*}
  y_i &= \beta + d_i \gamma + u_i,
\end{align*}
where $\beta=\theta$, $\gamma=b$. Now consider an $\ell_1$-penalized GLS regression estimator 
\begin{align*}
  (\hat\beta_{lasso,\lambda}',\hat\gamma_{lasso,\lambda})
  =\arg\min_{\beta,\gamma} \frac{1}{2} \|\tilde y-\tilde X\beta-\tilde z\gamma\|_2^2 + \lambda |\gamma|,
\end{align*}
where $\tilde y$, $\tilde z$,  and $\tilde X$ are transformed so that the observations are orthogonalized and standardized. 

\begin{theorem}\label{lasso_soft_threshold_theorem}
  Suppose that the lasso penalty $\lambda$ is set to equal to the adaptive soft-threshold (divided by $ \sigma_O$). Then the lasso regression coefficient estimator 
  \begin{align*}
  \hat\beta_{lasso,\lambda} = Y_{R,GMM} + \rho \sigma_U \delta_{S,\lambda \sigma_O}(T_O).
\end{align*}
 is the same as the soft-thresholding nearly adaptive estimator.
\end{theorem}
\begin{proof}
We first prove a general representation of the lasso regression coefficient estimator as a soft-thresholding estimator, and then we specialize the result to our setting. Consider a penalized regression estimator
\begin{align}\label{regularized_regression_eq}
  (\hat\beta_{\Pen,\lambda}',\hat\gamma_{\Pen,\lambda})
  =\arg\min_{\beta,\gamma} \frac{1}{2} \|y-X\beta-z\gamma\|_2^2 + \lambda \Pen(\gamma)
\end{align}
where $y$ and $Z$ are $n\times 1$ vectors and $X$ is a $n\times k$ matrix.
We use $P_X=X(X'X)^{-1}X'$ and $M_X=I-P_X$ to denote the projection onto the
column space of $X$ and onto its orthogonal complement. We are interested in the scalar parameter $\ell'\beta$ for some known vector
$\ell$ and wish to compare the estimator
$\ell'\hat\beta_{\Pen,\lambda}$ to estimators that are optimally adaptive or
constrained optimally adaptive for $\ell'\beta$ under a restriction on the bias
of the short regression estimator $\ell'\hat\beta_{\operatorname{short}}$ where
$\hat\beta_{\operatorname{short}}=(X'X)^{-1}X'y$.  

Note that standard regression algebra
immediately implies that $\hat\beta_{\Pen,\lambda}$ can be obtained by
regressing $y-z\hat\gamma_{\Pen,\lambda}$ on $X$, which gives
\begin{align}\label{beta_penalized_regression_as_a_function_of_gamma_eq}
  \ell'\hat\beta_{\Pen,\lambda}
  =\ell'(X'X)^{-1}X'(y-z\hat\gamma_{\Pen,\lambda})
  =\ell'\hat\beta_{\operatorname{short}} - \ell'(X'X)^{-1}X'z\hat\gamma_{\Pen,\lambda}.
\end{align}
To derive $\hat\gamma_{\Pen,\lambda}$, note that the objective in
(\ref{regularized_regression_eq}) can be written as
\begin{align*}
  \frac{1}{2}\|M_Xy-M_X z\gamma\|_2^2 + \frac{1}{2}\|P_X(y-z\gamma) - X\beta\|_2^2 + \lambda  \Pen(\gamma).
\end{align*}
Since the second term can be set to zero for any value of $\gamma$ by taking
$\beta=(X'X)^{-1}X'(y-z\gamma)$, and $\beta$ does not show up in the remaining
terms, it follows that this term can be ignored when optimizing
$\hat\gamma_{\Pen,\lambda}$.  Thus, $\hat\gamma_{\Pen,\lambda}$ minimizes
\begin{align*}
  \frac{1}{2}\|M_Xy-M_X z\gamma\|_2^2 + \lambda \Pen(\gamma).
\end{align*}

Consider the lasso case where $\Pen(\gamma)=|\gamma|$.
Taking FOCs gives
\begin{align*}
  &- z'M_X(y-z\gamma) + \lambda \operatorname{sign}(\gamma) = 0  \\
  &\Longleftrightarrow
  \gamma =\frac{z'M_Xy}{z'M_X z} - \frac{\lambda}{z'M_Xz} \operatorname{sign}(\gamma)
  =\hat\gamma_{\operatorname{long}} - \frac{\lambda}{z'M_Xz} \operatorname{sign}(\gamma)
\end{align*}
where $\operatorname{sign}(\gamma)$ is the set-valued function equal to the sign
of $\gamma$ when $\gamma$ is nonzero, and equal to $[-1,1]$ when $\gamma=0$.
There are three cases to consider.  First, if
$\hat\gamma_{\operatorname{long}}>\lambda/z'M_Xz$, then
$\operatorname{sign}(\gamma)=1$ so that
$\gamma=\hat\gamma_{\operatorname{long}}-\lambda/z'M_Xz$.  Second, if
$\hat\gamma_{\operatorname{long}}<-\lambda/z'M_Xz$, then
$\operatorname{sign}(\gamma)=-1$ so that
$\gamma=\hat\gamma_{\operatorname{long}}+\lambda/z'M_Xz$.  Finally, if
$\hat\gamma_{\operatorname{long}}\in [-\lambda/z'M_Xz,\lambda/z'M_Xz]$, then we
will run into a contradiction if $\gamma\ne 0$: $\gamma>0$ would imply
$\operatorname{sign}(\gamma)=1$ which would give
$\gamma=\hat\gamma_{\operatorname{long}}-\lambda/z'M_Xz\le 0$
and $\gamma<0$ would imply $\operatorname{sign}(\gamma)=-1$ which would give
$\gamma=\hat\gamma_{\operatorname{long}}+\lambda/z'M_Xz\ge 0$.  Thus, if
$\hat\gamma_{\operatorname{long}}\in [-\lambda/z'M_Xz,\lambda/z'M_Xz]$, we must
have $\gamma=0$.  It follows that the solution to the optimization problem is
given by
\begin{align*}
  \hat\gamma_{\Pen,\gamma}
  &=\begin{cases}
     0 & \text{when } |\hat\gamma_{\operatorname{long}}|\le |\lambda/z'M_Xz|    \\
     \hat\gamma_{\operatorname{long}} - \lambda/z'M_Xz  & \text{when } \hat\gamma_{\operatorname{long}}>\lambda/z'M_Xz  \\
     \hat\gamma_{\operatorname{long}} + \lambda/z'M_Xz  & \text{when } \hat\gamma_{\operatorname{long}}<\lambda/z'M_Xz  \\
   \end{cases}
\end{align*}
This is the soft-threshold estimator $\delta_{S,\lambda/z'M_Xz}(\hat\gamma_{\operatorname{long}})$ with cutoff $\lambda/z'M_Xz$.
Plugging this into (\ref{beta_penalized_regression_as_a_function_of_gamma_eq})
gives the penalized regression estimate for our parameter of interest as
\begin{align*}
  \ell'\hat\beta_{\Pen,\lambda}
  =\ell'\hat\beta_{\operatorname{short}} - \ell'(X'X)^{-1}X'z \cdot \delta_{S,\lambda/z'M_Xz}(\hat\gamma_{\operatorname{long}})
\end{align*}

Now apply the GLS transformation to the data as follows
\begin{align*}
  \tilde y = \left(
  \begin{array}{c}
    Y_{R,GMM}/\sigma_{R,GMM}  \\  T_{O}
  \end{array}
  \right)
  =
    \left(
    \begin{array}{cc}
      \frac{1}{\sigma_{R,GMM}} & 0  \\
      0 & \frac{1}{\sigma_O}
    \end{array}
    \right)
    \left(
    \begin{array}{cc}
      1 + \rho \frac{\sigma_U}{\sigma_O}  & -\rho \frac{\sigma_U}{\sigma_O}  \\
      - 1 & 1
    \end{array}
    \right)
  \left(
  \begin{array}{c}
    Y_U \\  Y_R
  \end{array}
  \right),
\end{align*}
\begin{align*}
  \tilde X
  =  \left(
    \begin{array}{cc}
      \frac{1}{\sigma_{R,GMM}} & 0  \\
      0 & \frac{1}{\sigma_O}
    \end{array}
    \right)
    \left(
    \begin{array}{cc}
      1 + \rho \frac{\sigma_U}{\sigma_O}  & -\rho \frac{\sigma_U}{\sigma_O}  \\
      - 1 & 1
    \end{array}
    \right)
    \left(
    \begin{array}{c}
      1 \\ 1
    \end{array}
  \right)
  =\left(
  \begin{array}{c}
    \frac{1}{\sigma_{R,GMM}}  \\
    0
  \end{array}
  \right)
\end{align*}
\begin{align*}
  \tilde z
  =  \left(
       \begin{array}{cc}
         \frac{1}{\sigma_{R,GMM}} & 0  \\
         0 & \frac{1}{\sigma_O}
       \end{array}
       \right)
       \left(
       \begin{array}{cc}
         1 + \rho \frac{\sigma_U}{\sigma_O}  & -\rho \frac{\sigma_U}{\sigma_O}  \\
         - 1 & 1
       \end{array}
       \right)
       \left(
       \begin{array}{c}
         0 \\ 1
       \end{array}
  \right)
  =\left(
  \begin{array}{c}
    -\frac{1}{\sigma_{R,GMM}}\cdot \rho\frac{\sigma_U}{\sigma_O}  \\
    \frac{1}{\sigma_O}
  \end{array}
  \right).
\end{align*}

The least squares
estimator of $\gamma$ is the minimum variance unbiased estimate for $\gamma=b$,
which is $\hat\gamma_{\operatorname{long}}=Y_O$.  The short regression estimator of $\beta$ in the transformed model is
$\hat\beta_{\operatorname{short}}=(\tilde X'\tilde X)^{-1}\tilde X' \tilde y=Y_{R,GMM}$.  Finally,
$(\tilde X'\tilde X)^{-1}\tilde X'\tilde z=\sigma^2_{R,GMM}\cdot
\frac{1}{\sigma_{R,GMM}}\cdot \frac{-1}{\sigma_{R,GMM}}\cdot
\frac{\rho \sigma_U}{\sigma_O}=
-\rho\frac{\sigma_U}{\sigma_O}$ and  $\tilde z' M_{\tilde X}\tilde z=1/\sigma^2_O$.  Thus, the GLS lasso
estimate is $Y_{R,GMM} + \rho \frac{\sigma_U}{\sigma_O} \delta_{S,\lambda \sigma^2_O}(Y_O).$
Note that soft-thresholding $Y_O$ at $\lambda \sigma^2_O$ is equivalent to soft
thresholding $T_O=Y_O/\sigma_O$ at $\lambda \sigma_O$ and multiplying
by $\sigma_O$.  Thus, we can also write the GLS lasso estimate as
  $Y_{R,GMM} + \rho \sigma_U \delta_{S,\lambda \sigma_O}(T_O).$
This is the same as the soft-thresholding nearly adaptive estimator, but with
$\lambda$ replaced by $\lambda\cdot \sigma_O$.

\end{proof}

\newpage
\setcounter{page}{1}
\renewcommand{\thesection}{Appendix \Alph{section}}
\renewcommand{\thesubsection}{\Alph{section}.\arabic{subsection}}

\begin{LARGE}
\begin{center}
Online Appendix to ``Adapting to Misspecification''
\end{center}
\end{LARGE}

\begin{large}
\begin{center}
Timothy B. Armstrong,
Patrick Kline
and
Liyang Sun
\end{center}
\end{large}

\begin{large}
\begin{center}
\date{September 2025}
\end{center}
\end{large}

\bigskip

\section{Additional details}
\setcounter{theorem}{0}
\renewcommand{\thetheorem}{\Alph{section}.\arabic{theorem}}
\setcounter{lemma}{0}
\renewcommand{\thelemma}{\Alph{section}.\arabic{lemma}}
\setcounter{table}{0}
\renewcommand{\thetable}{A\arabic{table}}
\setcounter{figure}{0}
\renewcommand{\thefigure}{A\arabic{figure}}

\subsection{Constrained adaptation}\label{constrained_adaptation_sec_appx}

The constrained adaptive estimator solves the problem
\begin{align}\label{constrained_adaptation_eq}
  A^*(\mathcal{B};\overline R)
  =
  \inf_{\estimator} \sup_{B\in\mathcal{B}} \frac{\Rmax (B,\estimator)}{\Rminimax(B)}
  \quad \text{s.t.} \quad \sup_{B\in\mathcal{B}} \Rmax (B,\estimator)\le \overline R.
\end{align}
We can rewrite this formulation as a weighted minimax problem similar to the one in Section \ref{adaptation_as_scaled_minimax_sec} by setting $t=\overline
R/A^*(\mathcal{B};\overline R)$ and considering the problem
\begin{align}\label{t_constrained_adaptation_eq}
  \inf_{\estimator} \sup_{B\in\mathcal{B}} \max\left\{ \frac{\Rmax (B,\estimator)}{\Rminimax(B)},  \frac{\Rmax (B,\estimator)}{t}\right\}
  =\inf_{\estimator} \sup_{B\in\mathcal{B}} \frac{\Rmax (B,\estimator)}{\min\left\{\Rminimax(B),t\right\}}.
\end{align}
Indeed, any solution to (\ref{constrained_adaptation_eq}) must also be a
solution to (\ref{t_constrained_adaptation_eq}) with $t=\overline
R/A^*(\mathcal{B};\overline R)$, since any decision function achieving a
strictly better value of (\ref{t_constrained_adaptation_eq}) would satisfy the
constraint in (\ref{constrained_adaptation_eq}) and achieve a strictly better
value of the objective in (\ref{constrained_adaptation_eq}).
Conversely, letting $\tilde A^*(t)$ be the value of
(\ref{t_constrained_adaptation_eq}), any solution to
(\ref{t_constrained_adaptation_eq}) will achieve the same value of the objective
(\ref{constrained_adaptation_eq}) and will satisfy the constraint for $\bar
R=t\cdot \tilde A^*(t)$.  In fact, this solution to
(\ref{t_constrained_adaptation_eq}) will also solve
(\ref{constrained_adaptation_eq}) for $\bar R=t\cdot \tilde A^*(t)$ so long as
this value of $\bar R$ is large enough to allow some scope for adaptation.

Arguing as in Section \ref{adaptation_as_scaled_minimax_sec}, we can write the
optimization problem (\ref{t_constrained_adaptation_eq}) as
\begin{align}\label{t_constrained_weighted_minimax_eq}
  &\inf_{\estimator} \sup_{(\theta,b)\in \cup_{B'\in\mathcal{B}}\mathcal{C}_{B'}} \tilde\omega(\theta,b,t) R(\theta,b,\estimator),  \\
  \text{where } &
                  \tilde\omega(\theta,b,t)=\left( \inf_{B\in\mathcal{B}\text{ s.t. }(\theta,b)\in\mathcal{C}_B} \min\left\{ \Rmax(B),t \right\} \right)^{-1}
                  =\max\left\{ \omega(\theta,b), 1/t \right\}  \nonumber
\end{align}
and $\omega(\theta,b)$ is given in Lemma \ref{adaptation_as_weighted_minimax_lemma} in Section \ref{adaptation_as_scaled_minimax_sec}.
Thus, we can solve (\ref{t_constrained_adaptation_eq}) by solving for the
minimax estimator under the loss function $(\theta,b,d)\mapsto \tilde
\omega(\theta,b,t)L(\theta,b,d)$.
Letting $A^*(t)$ be the optimized objective function, we can then solve
(\ref{constrained_adaptation_eq}) by finding a $t$ such that $\bar R=t\cdot
A^*(t)$.

We summarize these results in the following lemma.

\begin{lemma}\label{constrained_adaptation_lemma}
  Any solution to (\ref{constrained_adaptation_eq}) is also a solution to
  (\ref{t_constrained_weighted_minimax_eq}) with $t=\overline
R/A^*(\mathcal{B};\overline R)$.
  Conversely, let $\tilde A^*(t)$ denote the value of
  (\ref{t_constrained_weighted_minimax_eq}) and let
  $\tilde R(t)=\tilde A^*(t)\cdot t$.  If $\tilde
  R(t)>\inf_{\estimator}\sup_{B\in\mathcal{B}}\Rmax(B,\estimator)$ and
  $\inf_{B\in\mathcal{B}}\Rminimax(B)>0$, then
  $A^*(\mathcal{B};\tilde R(t))=\tilde A^*(t)$ and any solution to
  (\ref{t_constrained_weighted_minimax_eq}) is also a solution to
  (\ref{constrained_adaptation_eq}) with $\bar R=\tilde R(t)$.
\end{lemma}

  \begin{proof}
  The first statement is immediate from the arguments proceeding the statement of the lemma in Section \ref{constrained_adaptation_sec}.  For
  the second statement, let $\bar\delta$ be a decision rule with
  $\sup_{B\in\mathcal{B}}\Rmax(B,\bar\delta)<\tilde R(t)$.  Such a decision rule
  exists and satisfies
  $\sup_{B\in\mathcal{B}} \frac{\Rmax(B,\bar\delta)}{\Rminimax(B)}<\infty$
  by the assumptions of the lemma.
  Let $\bestimator^*_t$ be a solution to
  (\ref{t_constrained_adaptation_eq}).

  Suppose, to get a contradiction, that a decision $\delta'$ satisfies
the constraint in (\ref{constrained_adaptation_eq}) with $\bar R=\tilde R(t)$
and achieves a strictly better value of the objective than $\tilde A^*(t)$.  For
$\lambda\in (0,1)$, let $\delta_{\lambda}'$ be the randomized decision rule that
places probability $\lambda$ on $\bar\delta$ and probability $1-\lambda$ on
$\delta'$, independently of the data $Y$.
Note that
$\Rmax(B,\delta'_{\lambda})
=\sup_{(\theta,b)\in\mathcal{C}_{B}} R(\theta,b,\delta'_{\lambda})
=\sup_{(\theta,b)\in\mathcal{C}_{B}} \left[ \lambda R(\theta,b,\bar\delta)
  +(1-\lambda)R(\theta,b,\delta') \right]
\le \sup_{(\theta,b)\in\mathcal{C}_{B}} \lambda R(\theta,b,\bar\delta)
+\sup_{(\theta,b)\in\mathcal{C}_{B}} (1-\lambda)R(\theta,b,\delta')
=\lambda \Rmax(B,\bar \delta)+(1-\lambda)
\Rmax(B,\delta')$
so that, for $\lambda\in (0,1)$,
\begin{align*}
  \sup_{B\in\mathcal{B}} \Rmax(B,\delta_\lambda)
  \le \lambda \sup_{B\in\mathcal{B}} \Rmax(B,\bar\delta)
  + (1-\lambda )\sup_{B\in\mathcal{B}} \Rmax(B,\delta')
  < \tilde R(t)
  = \tilde A^*(t)\cdot t
\end{align*}
and
\begin{align*}
  \sup_{B\in\mathcal{B}} \frac{\Rmax(B,\delta_\lambda)}{\Rminimax(B)}
  \le \lambda \sup_{B\in\mathcal{B}} \frac{\Rmax(B,\bar\delta)}{\Rminimax(B)}
  + (1-\lambda ) \sup_{B\in\mathcal{B}} \frac{\Rmax(B,\delta')}{\Rminimax(B)}.
\end{align*}
Since $\sup_{B\in\mathcal{B}} \frac{\Rmax(B,\bar\delta)}{\Rminimax(B)}$ is
finite and $\frac{\sup_{B\in\mathcal{B}}
  \Rmax(B,\delta')}{\Rminimax(B)}<\tilde A^*(t)$, the above display is strictly less
than $\tilde A^*(t)$ for small enough $\lambda$.  Thus, for small enough
$\lambda$, the objective function in (\ref{t_constrained_weighted_minimax_eq})
evaluated at the decision function $\delta_{\lambda}$ evaluates to
\begin{align*}
  \max\left\{ \sup_{B\in\mathcal{B}} \frac{\Rmax(B,\delta_\lambda)}{\Rminimax(B)}, \sup_{B\in\mathcal{B}}\frac{\Rmax(B,\delta_\lambda)}{t} \right\}
  <\max\left\{ \tilde A^*(t), \tilde R(t)/t \right\}
  =\tilde A^*(t),
\end{align*}
a contradiction.
\end{proof}

\subsection{Numerical results on estimators as a function of $1-\rho^2$}\label{appdx:nearly_adaptive_approx}

In practice, it is common to use a fixed threshold of 1.96, which corresponds to a pre-test rule that switches between the unrestricted estimator and the GMM estimator based on the result of the specification test. Doing so leads to high level of worst-case adaptation regret especially when $\rho^2$ is close to one as shown in Figure~\ref{fig:penalty}. To minimize the worst-case adaptation regret, the adaptive hard-threshold estimator needs to use a threshold that would increase to infinity as $\rho^2$ gets closer to one.

\begin{figure}[h!]
\caption{\label{fig:penalty}  Worst case adaptation regret as function of relative efficiency}

\begin{centering}
\includegraphics[width=4.5in, trim={0 1.3cm 0 0.5cm}, clip]{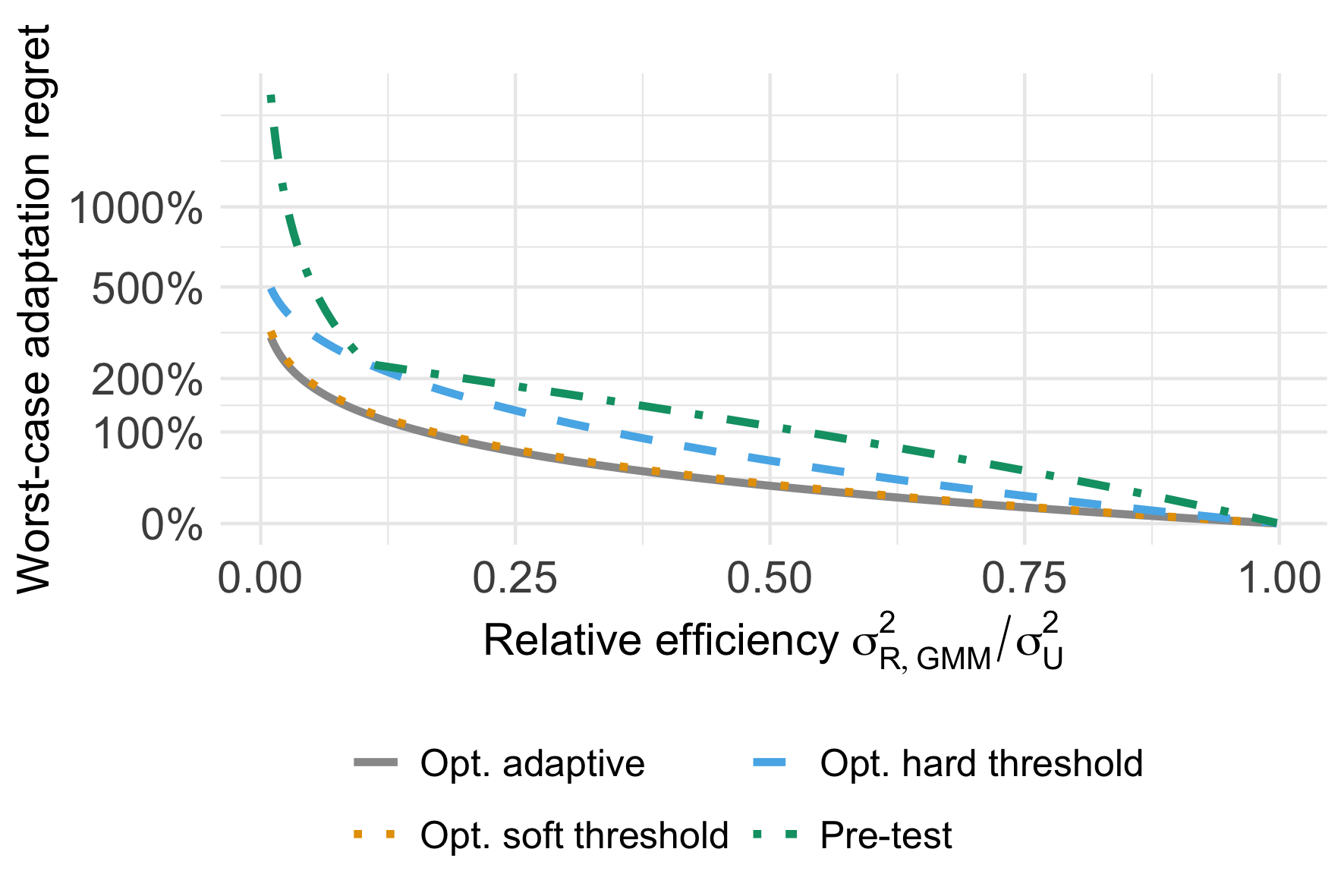}

\footnotesize{Notes: Vertical axis plots $(A_{\max}(\mathcal{B},\estimator)-1)\times100$ on $\log_{10}$ scale.}
\par\end{centering}
\end{figure}

\begin{figure}[h!]
\caption{\label{fig:max_risk} Worst case risk increase relative
to $Y_{U}$}

\begin{centering}
\includegraphics[width=4.5in, trim={0 1.3cm 0 0.5cm}, clip]{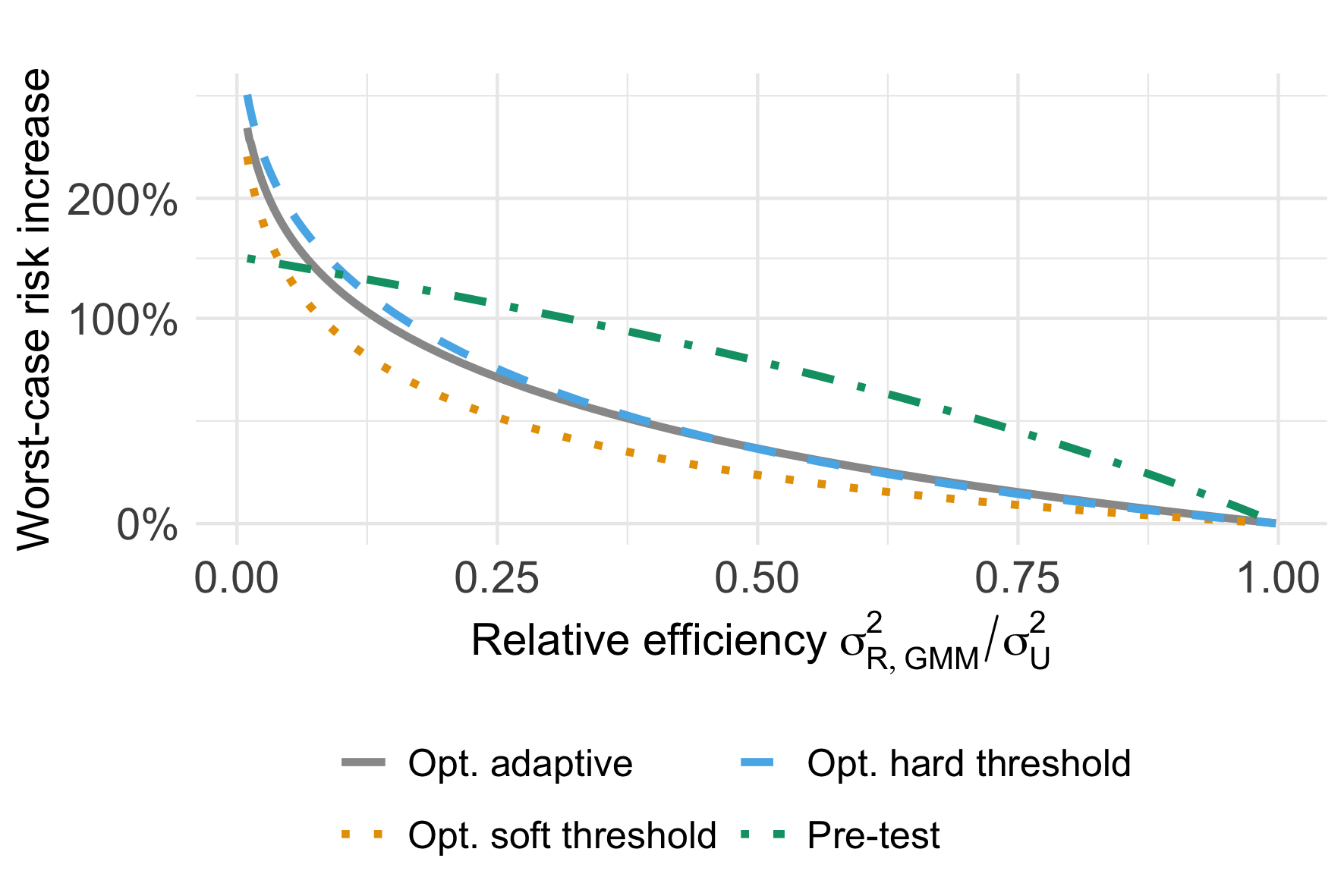}

\footnotesize{Notes: Vertical axis plots $(\Rmax(\infty,\estimator)-\sigma_{U})/\sigma_{U}\times100$ on $\log_{10}$ scale.}
\par\end{centering}
\end{figure}

A pre-test estimator utilizing a fixed threshold at 1.96 realizes its worst-case risk when the scaled bias $\tilde b$ is itself near the 1.96 threshold.  As shown in Figure~\ref{fig:max_risk}, the pre-test estimator tends to exhibit substantially greater worst-case risk than the class of adaptive estimators for most values of $\rho^2$.
As discussed in Section \ref{rho_to_one_section} below, adaptive estimators have large worst-case risk when $\rho^2$ is close to one.
The pre-test estimator has lower worst-case risk in these cases, due to the fixed threshold at $1.96$.

\subsection{Asymptotics as $|\rho|\to 1$}\label{rho_to_one_section}

This section considers the behavior of the worst-case adaptation regret as $|\rho|\to 1$ for the optimally adaptive estimator as well as for the hard and soft-thresholding estimators.
Recall that $1-\rho^2$ is equal to $\sigma_{R,GMM}^2/\sigma_U^2$, so that $|\rho|\to 1$ corresponds to the case where $\sigma_{R,GMM}^2/\sigma_U^2\to 0$.  It will be convenient to phrase our results in terms of $\rho^{-2}-1=(1-\rho^2)/\rho^2=(1+o(1))\cdot \sigma_{R,GMM}^2/\sigma_U^2$ as $|\rho|\to 1$.

Let $A(\delta,\rho)$ denote the worst-case adaptation regret of the
estimator given by (\ref{Bminimax_main_example_eq}) under the given value of
$\rho$, so that $A(\delta,\rho)$ returns the value of (\ref{adaptation_objective_main_example_eq}) with $\tilde\delta=\delta$.
We use $A^*(\rho)=\inf_{\delta}A(\delta,\rho)$ (where the infimum is over all estimators) to denote the loss of efficiency under
adaptation for the given
value of $\rho$.
Likewise, we denote by $A_{S}(\lambda,\rho)=A(\delta_{S,\lambda},\rho)$ and
$A_{H}(\lambda,\rho)=A(\delta_{H,\lambda},\rho)$ the
worst-case adaptation regret for
soft and hard-thresholding respectively with threshold $\lambda$, where $\delta_{S,\lambda}$ are $\delta_{H,\lambda}$ are defined in
Section \ref{nearly_adaptive_computation_sec}.
Finally, we use $A^*_S(\rho)=\inf_{\lambda} A_S(\lambda,\rho)$ and
$A^*_H(\rho)=\inf_{\lambda} A_H(\lambda,\rho)$ to denote the minimum worst-case
adaptation regret for soft and hard-thresholding respectively.

The following theorem characterizes the behavior of $A^*(\rho)$, $A^*_S(\rho)$ and
$A^*_H(\rho)$ as $|\rho|\to 1$.

\begin{theorem}\label{rho_to_one_limit_thm}
  We have
  \begin{align*}
    \lim_{|\rho|\uparrow 1} \frac{A^*(\rho)}{2\log (\rho^{-2}-1)^{-1}}
    =\lim_{|\rho|\uparrow 1} \frac{A^*_S(\rho)}{2\log (\rho^{-2}-1)^{-1}}
    =\lim_{|\rho|\uparrow 1} \frac{A^*_H(\rho)}{2\log (\rho^{-2}-1)^{-1}}
    =1.
  \end{align*}
\end{theorem}

In the remainder of this section, we prove Theorem \ref{rho_to_one_limit_thm}.
We split the proof into upper bounds (Section \ref{rho_to_one_upper_bound_sec})
and lower bounds (Section \ref{rho_to_one_lower_bound_sec}).  The lower bounds
in Section \ref{rho_to_one_lower_bound_sec} are essentially immediate from
results in \citetOnline{bickel_minimax_1983} for adapting to $B\in
\mathcal{B}=\{0,\infty\}$, whereas the upper bounds in Section
\ref{rho_to_one_upper_bound_sec} involve new arguments to deal with intermediate
values of $B$.

\subsubsection{Upper bounds}\label{rho_to_one_upper_bound_sec}

In this section, we show that $A^*_S(\rho)\le (1+o(1))2\log (\rho^{-2}-1)^{-1}$ and
$A^*_H(\rho)\le (1+o(1))2\log (\rho^{-2}-1)^{-1}$.  Since $A^*(\rho)$ is bounded from
above by both $A^*_S(\rho)$ and $A^*_H(\rho)$, this also implies $A^*(\rho)\le
(1+o(1))2\log (\rho^{-2}-1)^{-1}$.

Let $r_S(\lambda,t)=E_{T\sim N(\mu,1)}(\delta_{S,\lambda}(T)-\mu)^2$
and
$r_S(\lambda,t)=E_{T\sim N(\mu,1)}(\delta_{H,\lambda}(T)-\mu)^2$
denote the risk of soft and hard-thresholding.  Then
\begin{align*}
  A_{S}(\lambda,\rho)=\sup_{\mu\in\mathbb{R}} \frac{r_{S}(\lambda,\mu) + \rho^{-2} - 1}{\rBNM(|\mu|) + \rho^{-2}-1}
\end{align*}
and similarly for $A_{H}(\lambda,\rho)$.
We use the following upper bound for $r_H(\lambda,\mu)$ and $r_S(\lambda,\mu)$,
which follows immediately from results given in \citetOnline{johnstone_gaussian_2019}.

\begin{lemma}\label{thresholding_risk_upper_bound_lemma}
  There exists a constant $C$ such that, for $\lambda>C$, both
  $r_S(\lambda,\mu)$ and $r_H(\lambda,\mu)$ are bounded from above by $\bar
  r(\lambda,\mu)$ where
  \begin{align*}
    \bar r(\lambda,\mu)=
    \begin{cases}
      \min\left\{ \lambda \exp\left( -\lambda^2/2 \right) + 1.2\mu^2, 1+\mu^2 \right\}  &  |\mu|\le \lambda  \\
      1+\lambda^2  &  |\mu|>\lambda.
    \end{cases}
  \end{align*}
  
\end{lemma}
\begin{proof}
  The bound for $r_H(\lambda,\mu)$ follows from Lemma 8.5 in
  \citetOnline{johnstone_gaussian_2019} along with the bound $r_H(\lambda,0)\le
  \frac{2+\varepsilon}{\sqrt{2\pi}}\lambda \exp\left( -\lambda^2/2 \right)$
  which holds for any $\varepsilon>0$ for $\lambda$ large enough by (8.15) in
  \citetOnline{johnstone_gaussian_2019}.
  The bound for $r_L(\lambda,\mu)$ follows from Lemma 8.3 and (8.7) in \citetOnline{johnstone_gaussian_2019}.
\end{proof}

Let $\tilde\lambda_{\rho}=\sqrt{2\log (\rho^{-2}-1)^{-1}}$.  By Lemma
\ref{thresholding_risk_upper_bound_lemma}, $A^*_S(\rho)$ and $A^*_H(\rho)$ are,
for $(\rho^{-2}-1)^{-1}$ large enough,
bounded from above by the supremum over $\mu$ of
\begin{align}\label{thresholding_risk_upper_bound_eq}
  \frac{\bar r(\tilde\lambda_{\rho},\mu)+\rho^{-2}-1}{\rBNM(|\mu|)+\rho^{-2}-1}
\end{align}
Let $c(\rho)$ be such that $c(\rho)/\tilde \lambda_{\rho}\to 0$ and
$c(\rho)\to\infty$ as $|\rho|\uparrow 1$.  We bound
(\ref{thresholding_risk_upper_bound_eq}) separately for $|\mu|\le c(\rho)$ and
for $|\mu|\ge c(\rho)$.  For $|\mu|\le c(\rho)$, we use the bound $\rBNM(|\mu|)\ge .8\cdot \mu^2/(\mu^2+1)$
\citep{donoho_statistical_1994}, which gives an upper bound for
(\ref{thresholding_risk_upper_bound_eq}) of
\begin{align*}
  &\frac{\bar r(\tilde\lambda_{\rho},\mu)+\rho^{-2}-1}{.8\cdot \mu^2/(\mu^2+1)+\rho^{-2}-1}
  \le \frac{\sqrt{2\log (\rho^{-2}-1)^{-1}}\cdot (\rho^{-2}-1) + 1.2\mu^2+\rho^{-2}-1}{.8\cdot \mu^2/(\mu^2+1)+\rho^{-2}-1}  \\
  &\le \sqrt{2\log (\rho^{-2}-1)^{-1}}
    +(1.2/.8)\cdot (\mu^2+1)+1
  \le \sqrt{2\log (\rho^{-2}-1)^{-1}}
    +(1.2/.8)\cdot (c(\rho)^2+1)+1.
\end{align*}
As $|\rho|\uparrow 1$, this increases more slowly than $\log
(\rho^{-2}-1)^{-1}$.
For $|\mu|\ge c(\rho)$, we use the bound $\rBNM(|\mu|)\ge \rBNM(c(\rho))$ which
gives an upper bound for (\ref{thresholding_risk_upper_bound_eq}) of
\begin{align*}
  \frac{\bar r(\tilde\lambda_{\rho},\mu)+\rho^{-2}-1}{\rBNM(|c(\rho)|)+\rho^{-2}-1}
  \le \frac{\bar r(\tilde\lambda_{\rho},\mu)}{\rBNM(|c(\rho)|)} + 1
  \le \frac{1+\tilde \lambda_\rho^2}{\rBNM(|c(\rho)|)} + 1.
\end{align*}
As $|\rho|\uparrow 1$, $c(\rho)\to\infty$ and $\rBNM(|c(\rho)|)\to 1$, so that the
above display is equal to a $1+o(1)$ term times
$\tilde \lambda_\rho^2=2\log (\rho^{-2}-1)^{-1}$ as required.

\subsubsection{Lower bounds}\label{rho_to_one_lower_bound_sec}

In this section, we show that $A^*(\rho)\ge (1+o(1))2\log (\rho^{-2}-1)^{-1}$.
Since $A_S^*(\rho)$ and $A_H^*(\rho)$ are bounded from
below by $A^*(\rho)$ , this also implies $A_S^*(\rho)\ge
(1+o(1))2\log (\rho^{-2}-1)^{-1}$ and $A_H^*(\rho)\ge
(1+o(1))2\log (\rho^{-2}-1)^{-1}$.

Given an estimator $\delta(Y)$ of $\mu$ in the normal means problem $Y\sim
N(\mu,1)$, let $m(\delta)=E_{T\sim N(0,1)}\delta(Y)^2$ denote the risk at
$\mu=0$ and let
$M(\delta)=\sup_{\mu\in\mathbb{R}}E_{T\sim N(\mu,1)}(\delta(Y)-\mu)^2$ denote
worst-case risk.
The following lemma is immediate from \citetOnline[Theorem 4.1]{bickel_minimax_1983}.

\begin{lemma}[\citealtOnline{bickel_minimax_1983}, Theorem 4.1]\label{bickel_1983_lemma}
  For $t\in (0,1]$,
  let $\delta_t$ be an estimator that satisfies $m(\delta_t)\le 1-t$.  Then, as
  $t\uparrow 1$, $M(\delta_t)\ge (1+o(1))\cdot 2\log (1-t)$.
\end{lemma}

Using this result, we prove the following lemma, which gives a lower bound for
the worst-case adaptation regret and the worst-case risk of any estimator
achieving the upper bound in Section \ref{rho_to_one_upper_bound_sec}.  The
required lower bound $A^*(\rho)\ge (1+o(1))2\log (\rho^{-2}-1)^{-1}$ follows
from this result.

\begin{lemma}\label{rho_to_one_lower_bound_lemma}
  For $\rho\in (-1,1)$, let $\delta_\rho:\mathbb{R}\to\mathbb{R}$ be an
  estimator of $\mu$ in the normal means problem $Y\sim N(\mu,1)$.  Suppose that
  the worst-case adaptation regret $A(\delta_\rho,\rho)$ of the corresponding estimator
  (\ref{Bminimax_main_example_eq}) satisfies $A(\delta_\rho,\rho)\le
  (1+o(1))2\log (\rho^{-2}-1)^{-1}$ as $|\rho|\to 1$.
  Then the following results hold as $|\rho|\to 1$.
  \begin{itemize}
  \item[i.)] The worst-case risk of the corresponding estimator
    (\ref{Bminimax_main_example_eq}) is bounded from below by a $1+o(1)$ term
    times $2\Sigma_U \log (\rho^{-2}-1)^{-1}$

  \item[ii.)] $A(\delta_\rho,\rho)\ge (1+o(1))\cdot 2\log (\rho^{-2}-1)^{-1}$.
  \end{itemize}

\end{lemma}
\begin{proof}
  By the arguments Section \ref{details_and_proofs_appendix}, the
  worst-case risk of the estimator (\ref{Bminimax_main_example_eq}) with $\delta=\delta_\rho$ is given by
  $\Sigma_U\cdot\left[\rho^2\sup_{\mu}E_{T\sim N(\mu,1)}(\delta_\rho(T)-\mu)^2+
    1-\rho^2  \right]$.  As $|\rho|\uparrow 1$, this is bounded from below by a
  $1+o(1)$
  term times $\Sigma_U \sup_{\mu}E_{T\sim N(\mu,1)}(\delta_\rho(T)-\mu)^2$.
  Similarly, $A(\delta_\rho,\rho)$ is bounded from below by a $1+o(1)$ term
  times $\sup_{\mu}E_{T\sim N(\mu,1)}(\delta_\rho(T)-\mu)^2$ as $|\rho|\uparrow 1$.  Thus, it suffices
  to show that $\sup_{\mu}E_{T\sim N(\mu,1)}(\delta_\rho(T)-\mu)^2\ge
  (1+o(1))\cdot 2\log (\rho^{-2}-1)^{-1}$.

  To show this, note that it follows from plugging in $\tilde b=0$ to the objective in
  (\ref{adaptation_objective_main_example_eq}) that, for any $\varepsilon>0$, we
  have, for $|\rho|$ close enough to 1,
  \begin{align*}
    \frac{E_{T\sim N(0,1)}\delta_\rho(T)^2}{\rho^{-2}-1}
    \le A(\delta_\rho,\rho)\le (2+\varepsilon)\log (\rho^{-2}-1)^{-1}.
  \end{align*}
  Applying Lemma \ref{bickel_1983_lemma} with $1-t=(\rho^{-2}-1)\cdot
  (2+\varepsilon)\log (\rho^{-2}-1)^{-1}$, it follows that
  \begin{align*}
    &\sup_{\mu} E_{T\sim N(\mu,1)}(\delta_\rho(T)-\mu)^2
    \ge
    (1+o(1))\cdot 2\log \left[ (\rho^{-2}-1)\cdot
    (2+\varepsilon)\log (\rho^{-2}-1)^{-1} \right]  \\
    &=(1+o(1))\cdot \left[ 2\log (\rho^{-2}-1)
    +\log (2+\varepsilon)
    +\log \log (\rho^{-2}-1)^{-1} \right]
    =(1+o(1))\cdot 2\log (\rho^{-2}-1)
  \end{align*}
  as required.
  
\end{proof}

\section{Computational details}\label{appdx: lookup}

In this section, we provide additional details on our computation of the
adaptive estimator.

\subsection{Computing minimax estimators}

As shown in Sections \ref{adaptation_as_scaled_minimax_sec} and
\ref{main_example_computation_sec}, one can compute adaptive estimators by
solving a weighted minimax problem which, in our setting, can be further
simplified using invariance.  To solve these problems, we use the insight that
the minimax estimator can be characterized as a Bayes estimator for a
\emph{least favorable prior}.  We first give a brief review of this approach
before going into details for our setting.

Consider the generic problem of computing a minimax decision over the parameter
space $\mathcal{C}$ for a parameter $\vartheta$ under loss $\bar
L(\vartheta,\delta)$.  We use $E_\vartheta$ and $P_\vartheta$ to denote
expectation under $\vartheta$ and the probability distribution of the data $Y$
under $\vartheta$.
Letting $\pi$ denote a \emph{prior} distribution on $\mathcal{C}$, the \emph{Bayes risk} of $\delta$ is given by
\begin{align*}
\RBayes(\pi,\delta)=\int E_{\vartheta} \bar L(\vartheta,\delta(Y)) \, d\pi(\vartheta)
=\int \int \bar L(\vartheta,\delta(y)) \, dP_{\vartheta}(y) d\pi(\vartheta).
\end{align*}
The \emph{Bayes decision}, which we will denote $\delta_{\pi}^{\operatorname{Bayes}}$, optimizes $\RBayes(\pi,\delta)$ over $\delta$.  It can be computed by optimizing expected loss under the posterior distribution for $\vartheta$ taking $\pi$ as the prior.  Under squared error loss, the Bayes decision is the posterior mean.

$\RBayes(\pi,\delta)$ gives a lower bound for the worst-case risk of $\delta$ under $\mathcal{C}$ and $\RBayes(\pi,\delta_{\pi}^{\operatorname{Bayes}})$ gives a lower bound for the minimax risk.  Under certain conditions, a \emph{minimax theorem} applies, which tells us that this lower bound is in fact sharp.  In this case, letting $\Gamma$ denote the set of priors $\pi$ supported on $\mathcal{C}$, the minimax risk over $\mathcal{C}$ is given by
\begin{align*}
\min_{\delta}\max_{\pi\in\Gamma} \RBayes(\pi,\delta)
= \max_{\pi\in\Gamma}\min_{\delta} \RBayes(\pi,\delta)
= \max_{\pi\in\Gamma} \RBayes(\pi,\delta_{\pi}^{\operatorname{Bayes}}).
\end{align*}
The distribution $\pi$ that solves this maximization problem is called the \emph{least favorable prior}.  When the minimax theorem applies, the Bayes decision for this prior is the minimax decision over $\mathcal{C}$.

The expression $\RBayes(\pi,\delta_{\pi}^{\operatorname{Bayes}})$ is convex as a
function of $\pi$ if the set of possible decision functions is sufficiently
unrestricted and the set $\Gamma$ is convex. While one may need to allow
randomized decisions in general, the estimation problems we consider will be
such that the Bayes decision is nonrandomized. Thus, we can use convex
optimization software to compute the least favorable prior and minimax estimator
so long as we have a way of approximating $\pi$ with a finite dimensional object
that retains the convex structure of the problem.

In our setting, we use invariance arguments to obtain the objective function
(\ref{adaptation_objective_main_example_eq}), which is a minimax problem over
the unknown parameter $\tilde b=b/\sigma_O$ (the noncentrality parameter of the
overidentification statistic $T_O$).  We solve
(\ref{adaptation_objective_main_example_eq}), as well as the
bounded normal mean problem used to obtain the scaling in
(\ref{adaptation_objective_main_example_eq}), by solving for a least favorable
prior over $\tilde b$ using a finite dimensional approximation $\pi(\tilde
b_1),\ldots,\pi(\tilde b_J)$ to the prior over a grid of $J$ values of $\tilde b$.
The least favorable prior for $(\theta,b)$ is then given by a flat (improper)
prior for $\theta$ along with the corresponding prior for $\tilde b=b/\sigma_O$,
with the flat prior for $\theta$ following from invariance.
We now discuss the details of this approximation.

\subsection{Discrete approximation to estimators and risk function}

Operationally, discretizing the support of the random variable $T\in\mathcal{T}$
into $K$ points, finding an estimator $\bestimator(T)$ is equivalent
to finding a ``policy'' function $\bestimator\left(t\right):\mathcal{\mathcal{T}}\rightarrow\mathbb{R}$:
\[
\bestimator\left(t\right)=\sum_{k=1}^{K}\psi_{k}1\left\{ t=t_{k}\right\} .
\]
Hence, we can rewrite the risk of estimator $\bestimator(T)$ when $\ensuremath{T\sim N(b,1)}$
as
\begin{equation}
E_{T\sim N(b,1)}\left(\sum_{k=1}^{K}\psi_{k}1\left\{ T=t_{k}\right\} -b\right)^{2}.\label{eq:obj-bounded-normal-mean}
\end{equation}

Define $\mu_{kb}=\Pr_{T\sim N(b,1)}\left(T=t_{k}\right)$ as the probability
of falling into the $k$'th grid point given bias $b$, which can
be evaluated analytically via the following discrete approximation
to the normal distribution 
\begin{equation}
\mu_{kb}=\Phi\left(\left(t_{k}+t_{k+1}\right)/2-b\right)-\Phi\left(\left(t_{k}+t_{k-1}\right)/2-b\right), \label{eq:discrete density}
\end{equation}
where we define $t_{0}=-\infty$ and $t_{K+1}=\infty$, which ensures
that $\sum_{k=1}^{K}\mu_{kb}=1$. The discretized approximation to
the risk function (\ref{eq:obj-bounded-normal-mean}) is therefore
\begin{equation}
\sum_{k=1}^{K}\psi_{k}^{2}\mu_{kb}-2b\sum_{k=1}^{K}\psi_{k}\mu_{kb}+b^{2}.\label{eq:risk-approx}
\end{equation}

\subsection{Computing minimax risk in the bounded normal mean problem}\label{appdx: lookup_bnm}

We now provide details on how to compute the minimax risk $\rBNM(|\tilde{b}|)$
in the bounded normal mean problem, which allows us to easily compute
the $B$-minimax risk as described in \eqref{Rminimax_main_example_eq}
for each $B\in\mathcal{B}$.

By definition, the minimax risk $\rBNM(|\tilde{b}|)$ is the minimized
value of the following minimax problem 

\[
\min_{\delta}\max_{b\in[-|\tilde{b}|,|\tilde{b}|]}E_{T\sim N(b,1)}(\delta(T)-b)^{2}
\]
whose solution is the minimax estimator $\ensuremath{\deltaBNM\left(T;|\tilde{b}|\right)}$.
In particular, for each $|\tilde{b}|=B/\sigma_{O}\in\{0.1,0.2,\dots,9\}$
we calculate the minimax risk $\rBNM(|\tilde{b}|)$ following the
steps below. To compute the minimax risk function $\rBNM(|\tilde{b}|)$
for values of $|\tilde{b}|$ that are not included in the fine grid,
we rely on spline interpolation. 

\begin{enumerate}
\item Approximate the prior $\pi$ with the finite dimensional vector $\pi\in\Delta^{J}$, where the parameter space $[-|\tilde{b}|,|\tilde{b}|]$ is approximated
by an equally spaced grid of $b$ values spanning $[-|\tilde{b}|,|\tilde{b}|]$
with a step size of 0.05, totaling to $J$ grid values. Approximate
the conditional risk function as in (\ref{eq:risk-approx}), where
the support for $T\sim N(b,1)$ is approximated by an equally spaced
grid of $t$ values spanning $[-|\tilde{b}|-3,|\tilde{b}|+3]$ with
a step size of 0.1, totaling to $K$ grid values. The minimax problem
becomes
\begin{equation}
\max_{\pi\in\Delta^{J}}\min_{\left\{ \psi_{k}\right\} _{k=1}^{K}}\sum_{\ell=1}^{J}\pi_{\ell}\left(\sum_{k=1}^{K}\psi_{k}^{2}\mu_{kb_{\ell}}-2b_{\ell}\sum_{k=1}^{K}\psi_{k}\mu_{kb_{\ell}}+b_{\ell}^{2}\right).\label{eq:maxmin-bounded-normal-mean}
\end{equation}
\item The solution to the inner optimization yields the posterior mean $\psi_{k}^{*}\left(\pi\right)=\frac{\sum_{\ell=1}^{J}\pi_{\ell}\mu_{kb_{\ell}}b_{\ell}}{\sum_{\ell=1}^{J}\pi_{\ell}\mu_{kb_{\ell}}}$.
The outer problem is then 
\[
\max_{\pi\in\Delta^{J}}\sum_{\ell=1}^{J}\pi_{\ell}\left(\sum_{k=1}^{K}\left(\psi_{k}^{*}\left(\pi\right)\right)^{2}\mu_{kb_{\ell}}-2b_{\ell}\sum_{k=1}^{K}\psi_{k}^{*}\left(\pi\right)\mu_{kb_{\ell}}+b_{\ell}^{2}\right).
\]
\item Solve the outer problem for the least favorable prior $\pi^{\ast}$
based on sequential quadratic programming via MATLAB's \texttt{fmincon} routine.
The minimax estimator $\ensuremath{\deltaBNM\left(T;|\tilde{b}|\right)}$
is therefore $\sum_{k=1}^{K}\psi_{k}^{*}\left(\pi^{\ast}\right)1\left\{ t=t_{k}\right\} $
and the minimax risk $\rBNM(|\tilde{b}|)$ is the minimized value.
\end{enumerate}
Since the objective is concave in $\pi$ \citepOnline[it is the pointwise infimum over a set of linear functions; see][p. 81]{boyd_convex_2004}, we can check that the algorithm has found a global maximum by checking for a local maximum.

\subsection{Computing the optimally adaptive estimator for a given $\rho^{2}$}\label{appdx: lookup_opt}

As explained in the main text, the adaptive problem only depends on
$\Sigma$ through the correlation coefficient $\rho^{2}$. For a given
value of $\rho^{2}$, we use convex programming methods to solve for
the function $\tildedeltaadapt(t;\rho^2)$ based on the steps described
below.

\begin{enumerate}
\item Approximate the prior $\pi$ with the finite dimensional vector $\pi\in\Delta^{J}$,
where the parameter space for $b/\sigma_{O}$ is approximated by an equally spaced grid
of $\tilde{b}$ values spanning $[-9,9]$ with a step size of 0.025,
totaling to $J$ grid values. Approximate the conditional risk function
as in (\ref{eq:risk-approx}), where the support for $T\sim N(\tilde{b},1)$
is approximated by an equally spaced grid of $t$ values spanning
$[-12,12]$ with a step size of 0.05, totaling to $K$ grid values.
The adaptation problem \eqref{adaptation_objective_main_example_eq}
becomes
\begin{equation}
\max_{\pi\in\Delta^{J}}\min_{\left\{ \psi_{k}\right\} _{k=1}^{K}}\sum_{\ell=1}^{J}\pi_{\ell}\omega_{\ell}\left(\sum_{k=1}^{K}\psi_{k}^{2}\mu_{kb_{\ell}}-2b_{\ell}\sum_{k=1}^{K}\psi_{k}\mu_{kb_{\ell}}+b_{\ell}^{2}\right)+\rho^{-2}-1\label{eq:maxmin-discretized-adaptation}
\end{equation}
where $\omega_{\ell}=\left(\rBNM(|\tilde{b}_\ell|)+\rho^{-2}-1\right)^{-1}$using
output from the previous subsection.
\item The solution to the inner optimization yields $\psi_{k}^{*}\left(\pi\right)=\frac{\sum_{\ell=1}^{J}\pi_{\ell}\mu_{kb_{\ell}}\omega_{\ell}b_{\ell}}{\sum_{\ell=1}^{J}\pi_{\ell}\mu_{kb_{\ell}}\omega_{\ell}}$.
The outer problem is then 
\[
\max_{\pi\in\Delta^{J}}\sum_{\ell=1}^{J}\pi_{\ell}\omega_{\ell}\left(\sum_{k=1}^{K}\left(\psi_{k}^{*}\left(\pi\right)\right)^{2}\mu_{kb_{\ell}}-2b_{\ell}\sum_{k=1}^{K}\psi_{k}^{*}\left(\pi\right)\mu_{kb_{\ell}}+b_{\ell}^{2}\right)+\rho^{-2}-1.
\]
\item Solve the outer problem for the least favorable (adaptive) prior $\pi^{\ast}$
based on sequential quadratic programming via Matlab's fmincon routine.
The adaptive estimator $\tildedeltaadapt(t;\rho^2)$ is therefore $\sum_{k=1}^{K}\psi_{k}^{*}\left(\pi^{\ast}\right)1\left\{ t=t_{k}\right\} $.
The loss of efficiency under adaptation is the minimized value.
\end{enumerate}
As with the bounded normal mean problem, the objective is concave in $\pi$, so we can check that the algorithm has found a global maximum by checking for a local maximum.

This algorithm is a finite dimensional approximation to the optimization problem in Theorem \ref{main_example_thm_main_text}(\ref{main_example_thm_lea}).
While Theorem \ref{main_example_thm_main_text}(\ref{main_example_thm_lea}) does not formally show the existence of a solution to this infinite dimensional problem, we find that the algorithm reliably converges to a global maximum, and that the least favorable prior stabilizes as the number of gridpoints and range of the grid increase.
Based on this numerical finding, we conjecture that the minimax problem in Theorem \ref{main_example_thm_main_text}(\ref{main_example_thm_lea}) admits a least favorable prior, and that this solution can be approximated arbitrarily well using the our grid approach.

\subsection{Computing the optimally adaptive estimator based on the lookup table}\label{appdx: lookup_details}

To simplify the computation of the optimally adaptive estimator, we pre-calculate the adaptive estimates over an unequally spaced grid $\tanh([0,0.05,0.10,\dots,3])$ of correlation coefficients using the algorithm described above. As $\rho^2$ approaches one, the solution becomes sensitive to small changes in $\rho$. The uneven spacing of the $\rho$ grid allows for more accurate interpolation based on the simple pre-tabulated lookup table that we describe next. 

To rapidly obtain a final estimator $\tildedeltaadapt(T_O;\rho^2)$ for a given application, we conduct 2D interpolation across $\rho^2$ and $t$ values to tailor the adaptive estimates 
to the exact parameter values desired. For example, we obtain  $\tildedeltaadapt\left(T_{O};(-0.524)^2\right)$
based on spline interpolation at $\rho^2=(-0.524)^2$ together with
the observed test statistic $T_{O}$ based on the 2D grid of  $\rho^2$ and $t$ values.

\subsection{Computing the analytic adaptive estimators}\label{appdx:st}
 To find the analytic adaptive estimators in the class of ERM estimators, soft-thresholding estimators and hard-thresholding estimators, it suffices to  solve the two dimensional minimax problem in threshold $\lambda$ and scaled bias level $\tilde{b}$.  
We provide details for the claim in the main text that this two dimensional minimax problem can be easily solved even though the minimax theorem does not apply to these restricted classes of estimators.  To simplify the computation of the analytic adaptive estimator in practice, we pre-calculate
the adaptive thresholds $\lambda$ over an unequally spaced grid $\tanh([0,0.05,0.10,\dots,3])$
of correlation coefficients as explained above. To rapidly obtain
a final estimator, for example, soft-thresholding estimator  $\delta_{S,\lambda}\left(T_{O}\right)$ for
a given application, we conduct a spline interpolation across $\rho^{2}$
values to tailor the threshold to the exact parameter values desired.
For example, we obtain $\delta_{S,\lambda}\left(T_{O}\right)$
firstly based on spline interpolation at $\rho^{2}=(-0.524)^{2}$
to obtain the threshold $\lambda$, and then with the observed test
statistic $T_{O}$.

The derivation for soft and hard-thresholding is largely based on the following equality using moments of a truncated standard normal $X_i\mid a<X_i<b$.  Let $\phi(x)$ and $\Phi(x)$ denote the pdf and cdf of a standard normal distribution.  Then for any $ a < b $, we have 
\begin{eqnarray}
    \int_{a}^{b}x^{2}\phi(x)dx=\Phi\left(b\right)-\Phi\left(a\right)-\left(b\phi(b)-a\phi(a)\right).
\label{truncated_normal_lemma}
\end{eqnarray}

\subsubsection{Soft-thresholding}
Rewrite the soft-thresholding estimator as $\delta_{S,\lambda}\left(T_{O}\right)=\mathbf{1}\left\{ T_{O}>\lambda\right\} \left(T_{O}-\lambda\right)+\mathbf{1}\left\{ T_{O}<-\lambda\right\} \left(T_{O}+\lambda\right)$ and its risk function can be expressed as
\begin{eqnarray}
 &  & E_{T_{O}\sim N(\tilde{b},1))}\left(\delta_{S,\lambda}\left(T_{O}\right)-\tilde{b}\right)^{2}\nonumber \\
 & = & E_{T_{O}\sim N(\tilde{b},1)}\left(\mathbf{1}\left\{ T_{O}>\lambda\right\} \left(T_{O}-\lambda-\tilde{b}\right)+\mathbf{1}\left\{ T_{O}<-\lambda\right\} \left(T_{O}+\lambda-\tilde{b}\right)-\mathbf{1}\left\{ -\lambda<T_{O}<\lambda\right\} \tilde{b}\right)^{2}\nonumber \\
 & = & \tilde{b}^{2}\left(\Phi\left(\lambda-\tilde{b}\right)-\Phi\left(-\lambda-\tilde{b}\right)\right)+\int_{\lambda-\tilde{b}}^{\infty}\left(x-\lambda\right)^{2}\phi(x)dx+\int_{-\infty}^{-\lambda-\tilde{b}}\left(x+\lambda\right)^{2}\phi(x)dx\label{eq:st-integral}
\end{eqnarray}

The integrals in (\ref{eq:st-integral}) simplify to
\begin{align*}
 & \int_{\lambda-\tilde{b}}^{\infty}\left(x-\lambda\right)^{2}\phi(x)dx+\int_{-\infty}^{-\lambda-\tilde{b}}\left(x+\lambda\right)^{2}\phi(x)dx\\
= & \int_{\lambda-\tilde{b}}^{\infty}x^{2}\phi(x)dx+\int_{-\infty}^{-\lambda-\tilde{b}}x^{2}\phi(x)dx\\
 & -2\lambda\left(\int_{\lambda-\tilde{b}}^{\infty}x\phi(x)dx-\int_{-\infty}^{-\lambda-\tilde{b}}x\phi(x)dx\right)\\
 & +\lambda^{2}\left(1-\Phi\left(\lambda-\tilde{b}\right)+\Phi\left(-\lambda-\tilde{b}\right)\right)\\
= & 1-\Phi\left(\lambda-\tilde{b}\right)+\Phi\left(-\lambda-\tilde{b}\right)+\left((\lambda-\tilde{b})\phi(\lambda-\tilde{b})-(-\lambda-\tilde{b})\phi(-\lambda-\tilde{b})\right)\\
 & -2\lambda\left(\phi(\lambda-\tilde{b})+\phi(-\lambda-\tilde{b})\right)+\lambda^{2}\left(1-\Phi\left(\lambda-\tilde{b}\right)+\Phi\left(-\lambda-\tilde{b}\right)\right)
\end{align*}
where we use the fact that $\int_{\lambda-\tilde{b}}^{\infty}x^{2}\phi(x)dx+\int_{-\infty}^{-\lambda-\tilde{b}}x^{2}\phi(x)dx=\int_{-\infty}^{\infty}x^{2}\phi(x)dx-\int_{-\lambda-\tilde{b}}^{\lambda-\tilde{b}}x^{2}\phi(x)dx$ and Equation~\eqref{truncated_normal_lemma}.

The analytic adaptive objective function  
\[
\min_{\lambda}\max_{\tilde{b}}\frac{E_{T_{O}\sim N(\tilde{b},1))}\left(\delta_{S,\lambda}\left(T_{O}\right)-\tilde{b}\right)^{2}+\rho^{-2}-1}{\rBNM(|\tilde b|) + \rho^{-2}-1},
\]
can now be easily solved by Matlab's \texttt{fminimax} function when the risk function is evaluated based on the simplified expression derived above, and the parameter space for $\tilde b$ is approximated by an equally spaced grid
 values spanning $[-9,9]$ with a step size of 0.025.

\subsubsection{Hard-thresholding}
Similarly rewrite hard-thresholding as $\delta_{H,\lambda}\left(T_{O}\right)=\left(1-\mathbf{1}\left\{ -\lambda<T_{O}<\lambda\right\} \right)T_{O}$ and its risk function can be simplified due to  Equation~\eqref{truncated_normal_lemma}
\begin{eqnarray*}
 &  & E_{T_{O}\sim N(\tilde{b},1))}\left(\delta_{H,\lambda}\left(T_{O}\right)-\tilde{b}\right)^{2}\\
 & = & E_{T_{O}\sim N(\tilde{b},1)}\left(\left(1-\mathbf{1}\left\{ -\lambda<T_{O}<\lambda\right\} \right)\left(T_{O}-\tilde{b}\right)-\mathbf{1}\left\{ -\lambda<T_{O}<\lambda\right\} \tilde{b}\right)^{2}\\
 & = & \tilde{b}^{2}\left(\Phi\left(\lambda-\tilde{b}\right)-\Phi\left(-\lambda-\tilde{b}\right)\right)+\int_{-\infty}^{\infty}x^{2}\phi(x)dx-\int_{-\lambda-\tilde{b}}^{\lambda-\tilde{b}}x^{2}\phi(x)dx.
\end{eqnarray*} 
\subsubsection{Adaptive ERM}
For the adaptive ERM estimator $\delta_{ERM,\lambda}(T_{O})=\frac{T_{O}^{2}}{T_{O}^{2}+\lambda}\cdot T_{O}$, we evaluate the risk function based on $10^5$ simulations draws from $T_{O}\sim N(\tilde{b},1)$ and similarly optimize $\lambda$ for the analytic adaptive objective function.
\section{Pooling controls (LaLonde, 1986)} \label{sec:lalonde}
\citetOnline{lalonde1986evaluating} contrasted experimental estimates of the causal effects of job training derived from the National Supported Work (NSW) demonstration with econometric estimates derived from observational controls, concluding that the latter were highly sensitive to modeling choices. 
Subsequent work by \citeOnline{heckman1989choosing} argued that proper use of specification tests would have guarded against large biases in \citetOnline{lalonde1986evaluating}'s setting. An important limitation of the NSW experiment, however, is that its small sample size inhibits a precise assessment of the magnitude of selection bias associated with any given non-experimental estimator.  In what follows, we explore the prospects of improving experimental estimates of the NSW's impact on earnings by utilizing additional non-experimental control groups and adapting to the biases their inclusion engenders.

We consider three analysis samples differentiated by the origin of the untreated (``control'') observations. All three samples include the experimental NSW treatment group observations. In the first sample the untreated observations are given by the experimental NSW controls. In a second sample the controls come from \citeOnline{lalonde1986evaluating}'s observational ``CPS-1'' sample, as reconstructed by \citeOnline{dehejia1999causal}. In the third sample, the controls are a propensity score screened subsample of CPS-1. To estimate treatment effects in the samples with observational controls, we follow \citeOnline{angrist2009mostly} in fitting linear models for 1978 earnings to a treatment dummy, 1974 and 1975 earnings, a quadratic in age, years of schooling, a dummy for no degree, a race and ethnicity dummies, and a dummy for marriage status. The propensity score is generated by fitting a probit model of treatment status on the same covariates and dropping observations with predicted treatment probabilities outside of the interval $[0.1,0.9]$. 

Let $Y_U$ be the mean treatment / control contrast in the experimental NSW sample. We denote by $Y_{R1}$ the estimated coefficient on the treatment dummy in the linear model described above when the controls are drawn from the CPS-1 sample. Finally, $Y_{R2}$ gives the corresponding estimate obtained from the linear model when the controls come from the propensity score screened CPS-1 sample. We follow the applied literature in assuming trimming does not meaningfully change the estimand, a perspective that can be formalized by viewing the trimmed estimator as one realization of a sequence of estimators with trimming shares that decrease rapidly with the sample size \citepOnline{huber2013performance}.

Table~\ref{tab:minimax-training} reports point estimates from all three estimation approaches along with standard errors derived from the pairs bootstrap. The realizations of $(Y_{R1},Y_{R2})$ exactly reproduce those found in the last row of Table 3.3.3 of \citeOnline{angrist2009mostly} but the reported standard errors are somewhat larger due to our use of the bootstrap, which accounts both for heteroscedasticity and uncertainty in the propensity score screening procedure. The realization of $Y_U$ matches the point estimate reported in the first row of \citeOnline{angrist2009mostly}'s Table 3.3.3 but again exhibits a modestly larger standard error reflecting heteroscedasticity with respect to treatment status.

\begin{table}[h]
\caption{\label{tab:minimax-training}Estimates of the impact of
NSW job training on earnings.}
\begin{centering}
\begin{tabular}{cccccccc}
\hline 
 & $Y_{U}$ & $Y_{R1}$ & $Y_{R2}$ & $GMM_{2}$ & $GMM_{3}$ & Adaptive & Pre-test\tabularnewline
\hline 
\hline 
Estimate & 1794 & 794 & 1362 & 1629 & 1210 & 1596& 1629\tabularnewline
Std error & (668) & (617) & (741) & (619) & (595) & & \tabularnewline
Max Regret & 26\% & $\infty$ & $\infty$ & $\infty$ & $\infty$ & 7.8\% & 47.5\%\tabularnewline
Risk rel. to $Y_U$ &&&&&& \tabularnewline
when $b_{1}=0$ and $b_{2}=0$ & 1 & 0.853 & 1.23 & 0.858 & 0.793 & 0.855 & 0.80\tabularnewline
when $b_{1}\neq0$ and $b_{2}=0$ & 1 & $\infty$ & 1.23 & 0.858 & $\infty$ & 0.925 & 0.993\tabularnewline
when $b_{1}\neq0$ and $b_{2}\neq0$ & 1 & $\infty$ & $\infty$ & $\infty$ & $\infty$ & 1.078 &1.475 \tabularnewline

\hline 
\end{tabular}
\par\end{centering}
\footnotesize{Notes: Bootstrap standard errors in parentheses computed using 1,000 bootstrap samples. The $GMM_{2}$
estimate imposes $b_{2}=0$ only while the $GMM_{3}$ estimate imposes
$b_{1}=0$ and $b_{2}=0$. A $J$-test of the null $b_{1}=b_{2}=0$ motivating $GMM_3$
yields a p-value at 0.04. A corresponding test of the null $b_{2}=0$ motivating $GMM_2$
yields a p-value of 0.51. ``Risk rel. to $Y_U$'' gives worst case risk scaled by the risk (i.e. variance) of $Y_{U}$. ``Max regret'' refers to the worst case adaptation regret in percentage terms $(A_{\max}(\mathcal{B},\estimator)-1)\times100$.}

\end{table}

While the experimental mean contrast ($Y_U$) of \$1,794 is statistically distinguishable from zero at the 5\% level, considerable uncertainty remains about the magnitude of the average treatment effect of the NSW program on earnings. The propensity trimmed CPS-1 estimate lies closer to the experimental estimate than does the estimate from the untrimmed CPS-1 sample. However, the untrimmed estimate has a much smaller standard error than its trimmed analogue. Though the two restricted estimators are both derived from the CPS-1 sample, our bootstrap estimate of the correlation between them is only 0.75, revealing that each measure contains substantial independent information.

Combining the three estimators together via GMM, a procedure we denote $GMM_3$, yields roughly an 11\% reduction in standard errors relative to relying on $Y_U$ alone. However, the $J$-test associated with the $GMM_3$ procedure rejects the null hypothesis that the three estimators share the same probability limit at the 5\% level ($p=0.04)$. Combining only $Y_U$ and $Y_{R2}$ by GMM, a procedure we denote $GMM_2$, yields a standard error 7\% below that of $Y_U$ alone. The $J$-test associated with $GMM_2$ fails to reject the restriction that $Y_U$ and $Y_{R2}$ share a common probability limit ($p=0.51$). Hence, sequential pre-testing selects $GMM_2$.

Letting $b_1\equiv\mathbb{E}[Y_{R1}-\theta]$ and $b_2\equiv\mathbb{E}[Y_{R2}-\theta]$ our pre-tests reject the null that $b_1=b_2=0$ and fail to reject that $b_2=0$. However, it seems plausible that both restricted estimators suffer from some degree of bias. The adaptive estimator seeks to determine the magnitude of those biases and make the best possible use of the observational estimates. In adapting to misspecification, we operate under the assumption that $|b_1|\geq|b_2|$, which is in keeping with the common motivation of propensity score trimming as a tool for bias reduction \citepOnline[e.g.,][Section 3.3.3]{angrist2009mostly}. Denoting the bounds on $(|b_1|,|b_2|)$ by $(B_1,B_2)$, we adapt over the finite collection of bounds $\mathcal{B}=\{(0,0),(\infty,0),(\infty,\infty)\}$, the granular nature of which dramatically reduces the computational complexity of finding the optimally adaptive estimator. Note that the scenario $(B_1,B_2)=(0,\infty)$ has been ruled out by assumption, reflecting the belief that propensity score trimming reduces bias. See \ref{appdx: lalonde} for further details.

From Table~\ref{tab:minimax-training}, the multivariate adaptive estimator yields an estimated training effect of \$1,596: roughly two thirds of the way towards $Y_U$ from the efficient $GMM_3$ estimate. Hence, the observational evidence, while potentially quite biased, leads to a non-trivial (11\%) adjustment of our best estimate of the effect of NSW training away from the experimental benchmark. In Table~\ref{tab:bickel-nsw} we show that pairwise adaptation using only $Y_U$ and $Y_{R1}$ or only $Y_U$ and $Y_{R2}$ yields estimates much closer to $Y_U$. A kindred approach, which avoids completely discarding the information in either restricted estimator, is to combine $Y_{R1}$ and $Y_{R2}$ together via optimally weighted GMM and then adapt between $Y_U$ and the composite GMM estimate. As shown in Table~\ref{tab:bickel-nsw-composite}, this two step approach yields an estimate of \$1,624, extremely close to the multivariate adaptive estimate of \$1,596, but comes with substantially elevated worst case adaptation regret relative to a multivariate oracle who knows which pair of bounds in $\mathcal{B}$ prevails.

While the multivariate adaptive estimate of \$1,596 turns out to be very close to the pre-test estimate of \$1,629, the adaptive estimator's worst case adaptation regret of 7.8\% is substantially lower than that of the pre-test estimator, which exhibits a maximal regret of 47.5\%. The adaptive estimator achieves this advantage by equalizing the maximal adaptation regret across the three bias scenarios $\{(b_1=0,b_2=0), (b_1\neq0,b_2=0),(b_1\neq0,b_2\neq0)\}$ allowed by our specification of $\mathcal{B}$. When both restricted estimators are unbiased, the adaptive estimator yields a 14.5\% reduction in worst case risk relative to $Y_U$. However, an oracle that knows both restricted estimators are unbiased would choose to employ $GMM_3$, implying maximal adaptation regret of $0.855/0.793\approx 1.078$. When $Y_{R1}$ is biased, but $Y_{R2}$ is not, the adaptive estimator yields a 7.5\% reduction in worst case risk. An oracle that knows only $Y_{R1}$ is biased will rely on $GMM_2$, which yields worst case scaled risk of 0.858; hence, the worst case adaptation regret of not having employed $GMM_2$ in this scenario is $0.925/0.858\approx 1.078.$ Finally, when both restricted estimators are biased, the adaptive estimator can exhibit up to a 7.8\% increase in risk relative to $Y_U$. 

 The near oracle performance of the optimally adaptive estimator in this setting suggests it should prove attractive to researchers with a wide range of priors regarding the degree of selection bias present in the CPS-1 samples. Both the skeptic that believes the restricted estimators may be immensely biased and the optimist who believes the restricted estimators are exactly unbiased should face at most a 7.8\% increase in maximal risk from using the adaptive estimator. In contrast, an optimist could very well object to a proposal to rely on $Y_U$ alone, as doing so would raise risk by 26\% over employing $GMM_3$.

\section{Details of bivariate adaptation}\label{appdx: lalonde}

In \ref{sec:lalonde}, we report the results of adapting simultaneously to the bias in two restricted estimators when the bias spaces take a nested structure. Denoting the bounds on $(|b_1|,|b_2|)$ of the two restricted estimators by $(B_1,B_2)$, we adapt over the finite collection of bounds $\mathcal{B}=\{(0,0),(\infty,0),(\infty,\infty)\}$. Note that the scenario $(B_1,B_2)=(0,\infty)$ has been ruled out by assumption, reflecting the belief that propensity score trimming reduces bias.   
The minimax risk over each bias space $\mathcal{C}_{(B_1,B_2)}$ is therefore
\begin{equation}
R^{\ast}(\mathcal{C}_{(B_1,B_2)})=\begin{cases}
\Sigma_{U} & \text{ for }(B_1,B_2)=(\infty,\infty)\\
\Sigma_{U}-\Sigma_{UO,2}\Sigma_{O,2}^{-1}\Sigma_{UO,2} & \text{ for }(B_1,B_2)=(\infty,0)\\
\Sigma_{U}-\Sigma_{UO}\Sigma_{O}^{-1}\Sigma_{UO} & \text{ for }(B_1,B_2)=(0,0)
\end{cases}\label{eq:bickel p=00003D2}
\end{equation}
 Then $\delta(Y_{O})$ is the solution to the following problem

\[
\inf_{\delta}\max_{(B_1,B_2)\in\mathcal{B}}\frac{\max_{b\in\mathcal{C}_{(B_1,B_2)}}E_{Y_{O}\sim N(b,\Sigma_{O})}(\delta(Y_{O})-\Sigma_{UO}\Sigma_{O}^{-1}b)^{2}+\Sigma_{U}-\Sigma_{UO}\Sigma_{O}^{-1}\Sigma_{UO}}{R^{\ast}(\mathcal{C}_{(B_1,B_2)})}
\]
 Since the three spaces are nested, we can rewrite the adaptation
problem as

\[
\inf_{\delta}\sup_{b\in\mathbb{R}\times\mathbb{R}}\frac{E_{Y_{O}\sim N(b,\Sigma_{O})}(\delta(Y_{O})-\Sigma_{UO}\Sigma_{O}^{-1}b)^{2}+\Sigma_{U}-\Sigma_{UO}\Sigma_{O}^{-1}\Sigma_{UO}}{\tilde{R}(\tilde{\mathcal{S}}(b))}
\]

where the scaling is
\begin{equation}
\tilde{R}(\tilde{\mathcal{S}}(b))=\protect\begin{cases}
\Sigma_{U}-\Sigma_{UO}\Sigma_{O}^{-1}\Sigma_{UO} & \text{ if }b_{1}=b_{2}=0\protect\\
\Sigma_{U}-\Sigma_{UO,2}\Sigma_{O,2}^{-1}\Sigma_{UO,2} & \text{ if }b_{1}\protect\neq0,b_{2}=0\protect\\
\Sigma_{U} & \text{ if }b_{1}\protect\neq0,b_{2}\protect\neq0
\protect\end{cases}\label{eq:scaling p=00003D2}
\end{equation}

Given the high dimensionality of the adaptation problem, we use CVX instead of Matlab's fmincon to solve the scaled minimax problem.

\subsection{Pairwise adaptation}

For comparison with the trivariate adaptation estimates reported in the text, we also consider pairwise adaptation using only $Y_U$ and
$Y_{R1}$ or only $Y_U$ and $Y_{R2}$, keeping the bias spaces as before.  Specifically to adapt using only $Y_U$ and
$Y_{Rj}$, we consider an oracle  where the set $\mathcal{B}$ of bounds $B$ on the bias consists of the two elements $0$ and $\infty$.

\begin{table}[h]
\caption{\label{tab:bickel-nsw}Pairwise adaptive estimates}
\begin{centering}
\begin{tabular}{ccccccc}
\hline 
 & $Y_{U}$ & $Y_{R}$ & GMM & Adaptive & Soft-threshold & Pre-test\tabularnewline
\hline 
\hline 
CPS-1 untrimmed & 1794 & 794 & 1123 & 1659 & 1608 & 1794\tabularnewline
Std error & (668) & (617) & (600) &  &  & \tabularnewline
Rel. risk when $b=0$ & 1 & 0.85 & 0.81 & 0.863 & 0.869 & 0.895\tabularnewline
Rel. risk when $b\neq0$ & 1 & $\infty$ & $\infty$ & 1.071 & 1.078 & 1.539\tabularnewline
Max Regret & 24\% & $\infty$ & $\infty$ & 7.1\% & 7.8\% & 54\%\tabularnewline
Max Regret & 26\% & $\infty$ & $\infty$ & 24.8\% & 25.6\% & 79.4\%\tabularnewline
(rel. to multivariate)&&&&&&\tabularnewline
Threshold &  &  &  &  & 0.63 & 1.96\tabularnewline
\hline 
CPS-1 trimmed & 1794 & 1362 & 1629 & 1657 & 1638 & 1362\tabularnewline
Std error & (668) & (741) & (619) &  &  & \tabularnewline
Rel. risk when $b=0$ & 1 & 1.23 & 0.86 & 0.9 & 0.91 & 1.166\tabularnewline
Rel. risk when $b\neq0$ & 1 & $\infty$ & $\infty$ & 1.05 & 1.055 & 2.043\tabularnewline
    Max Regret & 16.5\% & $\infty$ & $\infty$ & 5\% & 5.5\% & 104\%\tabularnewline
Max Regret & 26\% & $\infty$ & $\infty$ & 13.6\% & 14.2\% & 104\%\tabularnewline
(rel. to multivariate)&&&&&&\tabularnewline
Threshold &  &  &  &  & 0.62 & 1.96\tabularnewline
\hline 
\end{tabular}
\par\end{centering}
\footnotesize{Notes: Bootstrap standard errors in parentheses computed using 1,000 bootstrap samples. In the top panel $Y_R$ corresponds to estimates using the untrimmed CPS-1 as controls, which are referred to as $Y_{R1}$ in the main text. In the bottom panel, $Y_R$ corresponds to estimates derived from the propensity score trimmed CPS-1 sample, which are referred to as $Y_{R2}$ in the main text. Adaptive estimates adapt pairwise between $Y_U$ and $Y_R$
within panel. If applicable,
the adaptive thresholds are reported. ``Max regret'' refers to the worst case adaptation regret in percentage terms $(A_{\max}(\mathcal{B},\estimator)-1)\times100$. ``Max Regret (rel. to multivariate)? refers to the worst case adaptation regret in terms of the multivariate oracle.  ``Rel. risk'' gives worst case risk scaled by the risk (i.e. variance) of $Y_{U}$. The correlation between $Y_U$ and $Y_{Rj}-Y_{U}$
is -0.44 in the top panel and -0.38 in the bottom panel.}

\end{table}
 
Table~\ref{tab:bickel-nsw} shows that pairwise adaptation produces estimates much closer to $Y_U$ than the multivariate adaptive estimate.  While pairwise adaptive estimates both incur smaller adaptation regret, the efficiency gain when the model is correct is smaller than with the multivariate adaptive estimate.

\subsection{Bivariate adaptation with GMM composite}

\begin{table}[h!]
\caption{\label{tab:bickel-nsw-composite}Adapting pairwise with GMM composite}
\begin{centering}
\begin{tabular}{ccccccc}
\hline 
 & $Y_{U}$ & $Y_{\text{comp}}$ & GMM & Adaptive & Soft-threshold & Pre-test\tabularnewline
\hline 
\hline 
Estimate & 1794 & 882 & 1173 & 1624 & 1602 & 1794\tabularnewline
Std error & (668) & (612) & (595) &  &  & \tabularnewline
Max Regret & 26\% & $\infty$ & $\infty$ & 8\% & 8.3\% & 56\%\tabularnewline
Max Regret  & 26\% & $\infty$ & $\infty$ & 25.4\% & 26.3\% & 81.5\%\tabularnewline
(rel. to multivariate)&&&&&&\tabularnewline
Threshold &  &  & $\infty$ &  & 0.64 & 1.96\tabularnewline
\hline 
\end{tabular}
\par\end{centering}
\footnotesize{Notes: Adaptive estimates for the impact of
job training, adapting to $B_{\text{comp}}\in\{0,\infty\}$, which
is the bound on the bias of the composite estimator $Y_{\text{comp}}=\arg\min_{\theta}(Y_{R}-\theta)'\Sigma_{R}^{-1}(Y_{R}-\theta).$ GMM combines $Y_{\text{comp}}$ and $Y_U$ optimally under the assumption that $Y_{\text{comp}}$ is unbiased.
If applicable, the adaptive thresholds are reported. ``Max regret'' refers to the worst case adaptation regret in percentage terms $(A_{\max}(\mathcal{B},\estimator)-1)\times100$.  ``Max Regret (rel. to multivariate)'' refers to the worst case adaptation regret relative to the multivariate oracle in \eqref{eq:bickel p=00003D2}.  The correlation coefficient between $Y_U$ and $Y_{\text{comp}}-Y_{U}$ is -0.45. }

\end{table}

For another comparison with the trivariate adaptation estimates reported in the text, we also consider combining  
$Y_{R1}$ and $Y_{R2}$ first via optimally weighted GMM, which is a composite of the two $Y_{\text{comp}}$.  We then adapt between $Y_U$ and $Y_{\text{comp}}$. The bias space is now also a composite of the two-dimensional bias space $\mathcal{C}_{(B_1,B_2)}$, and we consider an oracle  where the set $\mathcal{B}$ of bounds $B$ on the bias consists of the two elements $0$ and $\infty$.

Table~\ref{tab:bickel-nsw-composite} shows that composite adaptation produces estimates very similar to the multivariate adaptive estimate.  The adaptation regret relative to an oracle who knows a bound on the bias of composite is also small.  However, for a fair comparison with multivariate adaptation, one should compare its efficiency loss relative to the multivariate oracle with minimax risk specified in \eqref{eq:bickel p=00003D2}. This notion of worst case regret is substantially higher at 25\% because bivariate adaptation against the GMM composite cannot leverage the nested structure of the multivariate parameter space $\mathcal{B}$. 

\bibliographystyleOnline{chicago}
\bibliographyOnline{references_low_dim_adapt,references_empirical,additional_references}
\end{document}